\documentclass{lmcs} %%% last changed 2014-08-20

\keywords{modal logic, finite model theory, preservation theorems, semantic characterizations, 
correspondence theory}

\usepackage{hyperref}
\usepackage{tikz}
\usetikzlibrary{arrows.meta}
\usepackage{caption}
\newcommand{\dist}{\mathrm{dist}}

\usepackage{amssymb}
\theoremstyle{plain} %\crefname{satz}{Satz}{S\"atze}
\begin{document}

% If the title is longer than 55 characters, then specify a shorter running title as the optional argument to \title. The running title should be roughyl at most 55 characters:
\title[Modal Definability and Preservation in the Finite]{When do modal definability and preservation theorems transfer to the finite?}
% \thanks{thanks, optional.}	%optional

% affiliations are numbered automatically with a, b, c (see below)
% use the optional argument to indicate the affiliation(s) of each author
% omit the argument if there is only one author, or only one affiliation
\author[J.~van Benthem]{Johan van Benthem~\lmcsorcid{0000-0002-7048-785X}}[a,b,c]
\author[B.~ten Cate]{Balder ten Cate~\lmcsorcid{0000-0002-2538-5846}}[a]
\author[X.~Yang]{Xi Yang~\lmcsorcid{0009-0003-5576-311X}}[a,b]

% affiliation 1 (automatically numbered a)
\address{ILLC, University of Amsterdam, The Netherlands}	%optional
% write emails for all authors having that affiliation
\email{j.vanbenthem@uva.nl, b.d.tencate@uva.nl, x.yang@uva.nl}  %optional

% affiliation 2 (automatically numbered b)
\address{The Tsinghua-UvA JRC for Logic, Department of Philosophy, Tsinghua University, China}

% affiliation 3
\address{Department of Philosophy, Stanford University, USA}

%% etc.

%% required for running head on odd and even pages, use suitable
%% abbreviations in case of long titles and many authors:

%%%%%%%%%%%%%%%%%%%%%%%%%%%%%%%%%%%%%%%%%%%%%%%%%%%%%%%%%%%%%%%%%%%%%%%%%%%

%% the abstract has to PRECEDE the command \maketitle:
%% be sure not to issue the \maketitle command twice!

\begin{abstract}
  \noindent We study which classic modal definability and preservation results survive when attention is restricted to finite structures, where many first-order transfer theorems are known to break down. Several semantic characterizations for modal formula classes survive the passage to the finite, while a number of first-order preservation theorems for basic frame operations fail. Our main positive result is that the Bisimulation Safety Theorem does transfer to finite structures. We also discuss computability aspects, and analogues in the finite for the Goldblatt–Thomason theorem and for modal correspondence theory.
\end{abstract}

\maketitle

\section{Introduction}

The model theory of first-order logic exhibits a striking divergence between the infinite and the finite. Many classic theorems---central to the semantic and structural understanding of first-order logic---fail when attention is restricted to finite structures \cite{Tait59,AjtaiGurevich87}. For instance, compactness breaks down, the \L o\'s--Tarski preservation theorem no longer characterizes the existential and universal fragments, and Craig interpolation fails once we move to the finite. Against this backdrop of failure for classical metatheory, it is all the more remarkable that a few deep results \emph{do} survive the transition to the finite. A notable case in point is the van Benthem Characterization Theorem~\cite{vanBenthem1985}, which identifies modal logic as the bisimulation-invariant fragment of first-order logic. It was shown by Rosen \cite{rosen-thesis} that this remains valid over finite structures. A second notable example is Rossman's Homomorphism Preservation Theorem \cite{rossman2008homomorphism}, establishing that the positive existential fragment of first-order logic is precisely the class of formulas preserved under (not necessarily surjective) homomorphisms, even in the finite. A long standing open problem remains whether, for relational signatures, the universal Horn fragment is characterized by closure under (finite) products and induced substructures in the finite~\cite{AlechinaGurevich1997,Ham2023}.

Within the more restricted setting of \emph{modal logic}, the situation is more hospitable. The basic modal language, interpreted over the class of all Kripke models,
has the finite model property. This gives rise to a natural heuristic: results that hold in the setting of arbitrary models for modal logic may plausibly carry over to the finite, even when their first-order analogues fail in the finite. For example, it follows quite directly from the finite model property that the Craig interpolation theorem for modal logic  transfers to the finite, which also implies that the Beth definability theorem  remains true in the finite.

In this paper we explore several complementary themes. 
\begin{enumerate}
    \item We review some basic model-theoretic preservation results concerning syntactic fragments of modal logic (interpreted over the class of all Kripke models) and show how they transfer to the finite. In particular, we show how the known syntactic characterizations of the semantic properties of complete additivity and continuity transfer.
    \item In the first-order setting, we study transfer failures of  preservation theorems for basic operations on frames such as taking generated subframes, disjoint unions, bounded morphic images, and ultrafilter extensions. As our main novel contribution, we show a positive result: the Bisimulation Safety Theorem, which characterizes the first-order operations on binary relations that are safe for bisimulation, carries over to the finite.  
    \item We briefly look into issues of computational complexity for testing preservation properties, clarifying the following noteworthy contrast.  While  preservation properties for the modal language are decidable, those for first-order logic tend to be undecidable.
    \item Next, we look at the celebrated Goldblatt-Thomason Theorem about modally definable frame classes and give a version that holds in the finite, while identifying some further open problems.
    \item Finally, we look at the related theme of which known results in modal correspondence theory transfer to the finite. In particular, we show that there are still non-first-order modally definable frame properties in the finite. More generally, we consider a hierarchy of correspondence languages with respect to frame definability, consisting of extensions of first-order logic, and we show that this hierarchy survives also 
    over finite frames, under standard complexity-theoretic assumptions.
    
\end{enumerate}

\section{The modal logic setting}
\label{sec:modal}

In this section, we discuss four known examples of preservation theorems for modal formulas and show that they transfer to the finite. 
In the first two cases, the arguments are very easy and can be considered folklore. The third example, \emph{complete additivity}, requires a more involved proof (and it will play an important role in the next section, where we will use it to prove that the Bisimulation Safety Theorem transfers to the finite). A recurring theme will be that we make use of the well known fact that 
modal logic has the finite model property: 

\begin{thm}[Finite Model Property]
A modal formula is valid over arbitrary Kripke models if and only if  it is valid over finite Kripke models.
\end{thm}

Note: we assume familiarity with modal logic. We refer to~\cite{BdRV} for a general introduction. 
Throughout this section, we will 
consider \emph{multi-modal} formulas, i.e., 
modal formulas that may contain 
not just one pair of modal operators $\Diamond, \Box$, but a family $\langle a\rangle, [a]$ for $a\in A$. These are then interpreted on 
Kripke models of the form $(W,(R_a)_{a\in A},V)$. 
In other words, our setting corresponds to that of the basic multi-modal logic \textbf{K}.
The fact that we allow multiple modalities does not affect the results or the 
proofs in any way, but it facilitates the use (in Section~\ref{sec:FO}) of
one of the results stated in the current section.

We will now consider four important semantic properties that modal formulas can have.  

\subsection{Monotonicity}

A modal  formula is said to be \emph{monotone} in a propositional variable
$p$ if its truth is preserved when enlarging the valuation of $p$. 
Formally, $\phi$ is monotone in $p$ if, for all Kripke models $M=(W,(R_a)_{a\in A},V)$ and $w\in W$, $M,w\models\phi$ implies that $M[p\mapsto X],w\models\phi$ for all sets $X\subseteq W$ with $V(p)\subseteq X$. 

\begin{thmC}[{\cite[Theorem 3.11]{Kurtonina1997:simulating}}, cf.~also~{\cite[Appendix B]{Kurtonina1999:expressiveness}}]
\label{thm:monotonicity}
   For all modal formulas $\phi$ and propositional variables $p$, the following are equivalent:
   \footnote{Technically speaking, the cited result applies to formulas that are monotone in \emph{all} propositional variables rather than in a specific one. In terms of the proof, this does not make an essential difference.}
\begin{enumerate}
    \item $\phi$ is monotone in $p$, \vspace{1mm}
    \item $\phi$ is equivalent to a modal formula $\psi$ that is positive in $p$.
\end{enumerate}
\end{thmC}

Here, we say that a modal formula $\psi$ is \emph{positive in $p$} if it is in negation normal form and does not contain any negated occurrences of $p$.

It turns out that a formula $\phi$'s being monotone in a propositional variable $p$ is correlated with the validity of an associated modal formula $\phi'$, both in the finite and in the infinite:
\begin{prop}\label{prop:monotone-validity}
   For all modal formulas $\phi$ and propositional variables $p$, the following statements are equivalent:
   \begin{enumerate}
       \item $\phi$ is monotone in $p$, \vspace{1mm}
       \item The following validity holds: $\models\phi\to \phi[p/(p\lor q)]$.
   \end{enumerate}
    with $q$ a fresh propositional variable. This equivalence holds both over arbitrary structures and in the finite.
\end{prop}

    Indeed, clearly, any counterexample to monotonicity provides a counterexample for the given validity, and vice versa. 

    In light of the finite model property, it now follows that Theorem~\ref{thm:monotonicity} holds in the finite: 
    if $\phi$ is monotone in $p$ over finite structures, then, by Proposition~\ref{prop:monotone-validity}, $\phi\to\phi[p/(p\lor q)]$ is valid 
    in the finite. Hence, by the finite model property, it is valid over arbitrary structures. Applying Proposition~\ref{prop:monotone-validity} once more, we derive that
    $\phi$ is monotone in $p$ over arbitrary structures. Therefore, by Theorem~\ref{thm:monotonicity}, 
    $\phi$ is equivalent over arbitrary structures, and hence also over finite structures, to a formula of the required syntactic shape. 
    
\subsection{Preservation under induced substructures}

 We say that a modal 
formula $\phi$ is \emph{preserved under induced substructures} if $M,w\models\phi$  implies
$M| X,w\models\phi$, where 
$X$ is any subset of the domain that includes $w$,
and $M| X$ is the restriction of $M$ to $X$.
The following preservation theorem is known for
preservation under induced substructures, cf.~also its history in \cite{NNIL}.

\begin{thmC}[{\cite[Theorem 6.6.5]{deRijkePhD}}]
\label{thm:induced}
   For all modal formulas $\phi$, the following are equivalent:
\begin{enumerate}
    \item $\phi$ is preserved under induced substructures,\vspace{1mm}
    \item $\phi$ is equivalent to a modal formula $\psi$ in negation normal form not containing any modal $\Diamond$ operator.
\end{enumerate}
\end{thmC}

The fact that a modal formula $\phi$ is preserved under induced substructures again correlates with validity of an associated  modal formula $\phi'$.
We first need to introduce some notation.
For a modal formula $\phi$ and a propositional variable $p$, 
let 
$\phi^p$ be the result of \emph{relativizing} $\phi$ \emph{to} $p$
    (i.e., replacing every $\langle a\rangle$ by $\langle a\rangle(p\land\cdots)$ and replacing every $[a]$ by $[a](p\to \cdots)$).

\begin{prop}\label{prop:lostarski-validity}
   For all modal formulas $\phi$, the following are equivalent:
   \begin{enumerate}
       \item $\phi$ is preserved under induced substructures, \vspace{1mm}
       \item The following validity holds: $\models (p\land\phi)\to \phi^p$.
   \end{enumerate}
    where $p$ is a fresh propositional variable. This equivalence holds both over arbitrary structures and in the finite.
\end{prop}

Indeed, any (finite) counterexample to preservation under induced substructures provides a (finite) counterexample for the given validity, and vice versa.

It again follows that Theorem~\ref{thm:induced} holds also in the finite, using the same line of argument as in the case of monotonicity.

\vspace{-1ex}
\subsection{Complete additivity}

    A modal formula $\phi$ is \emph{completely additive} in a propositional variable $p$
    if $\phi$ is monotone in $p$ and whenever
    $M,w\models\phi$, then there is an element
    $v\in V(p)$ such that $M[p\mapsto \{v\}],w\models\phi$. %\color{red}{{\bf  Programming Operations that are Safe for Bisimulation, Studia Logica 60:2 (Logic Colloquium. Clermont–Ferrand 1994), 311–330. }}\color{black}

\begin{thmC}[\cite{vBenthem1998:bisimulation}]
\label{thm:complete-additivity-pres}
    A modal formula $\phi$ is completely additive in a propositional variable $p$ if and only if $\phi$ is 
    equivalent to a finite disjunction of
    formulas $\theta$ that can each be 
    generated by the following recursive grammar:
    \[ \theta ~~:=~~ p\land\phi \mid \phi\land \langle a\rangle \theta\]
    where $\phi$ is a $p$-free modal formula and $a\in A$.
\end{thmC}

The fact that a modal formula $\phi$ is completely additive in a propositional variable $p$ can again be correlated to the validity of an associated modal formula $\phi'$, both in the finite and in the infinite. This time, the proof, however, is less straightforward.

\begin{prop} \label{prop:additivity-finite}
   For all modal formulas $\phi$ and propositional variables $p$, the following are equivalent, where  $q_1, q_2$  are fresh propositional variables:
   \begin{enumerate}
       \item $\phi$ is completely additive in $p$, \vspace{1mm}
       \item The following two validities hold: $\models\phi[p/q_1\lor q_2] \leftrightarrow (\phi[p/q_1]\lor \phi[p/q_2])$ and $\models\neg\phi[p/\bot]$.
   \end{enumerate}
    The equivalence holds  both over arbitrary structures and  in the finite.
\end{prop}

\begin{proof}
    From 1 to 2: Suppose $\phi$ is completely additive in $p$.
    In particular, this means that 
    $\phi$ is monotone in $p$.
    It follows from the monotonicity
    that $\models (\phi[p/q_1]\lor\phi[p/q_2])\to \phi[p/q_1\lor q_2]$. 
    Next, suppose $M,w\models\phi[p/q_1\lor q_2]$.
    Since $\phi$ is completely additive in $p$, it follows that $M[p\mapsto \{v\}],w\models\phi$ for some $v\in V(p_1)\cup V(p_2)$. If $w\in V(p_1)$, it follows by monotonicity that
    $M,w\models\phi[p/p_1]$, while, if
    $w\in V(p_2)$, it follows by monotonicity that $M,w\models\phi[p/p_2]$. It also follows directly from the complete additivity property that $\models\neg\phi[p/\bot]$. The same argument applies in the finite.
         
    From 2 to 1: 
    assume that the stated validities hold. 
    We will proceed in two steps. First, we will show that $\phi$ is completely additive over finite structures. We prove this by contraposition: let $(M,w)$ be any
    finite counterexample to the complete
    additivity of $\phi$ in $p$.
    In particular, $M=(W,(R_a)_{a\in A},V)$ is a finite Kripke structure, $M,w\models\phi(p)$
    and one of the two conditions holds:\vspace{1mm}
    \begin{enumerate}
        \item There are sets $X\subseteq Y\subseteq W$ such that
        $(M[p\mapsto X],w)\models\phi$
        and $(M[p\mapsto Y],w)\not\models\phi$, or \vspace{1mm}
        \item There is a set $X\subseteq W$
        such that $(M[p\mapsto X],w)\models\phi$ and there is no 
        $x\in X$ such that $(M[p\mapsto \{x\}],w)\models\phi$.
    \end{enumerate}
    In the first case, we have
    $(M[q_1\mapsto X, q_2\mapsto Y],w)\not\models  (\phi[p/q_1]\lor \phi[p/q_2]) \to \phi[p/q_1\lor q_2]$.
    In the second case, if $X=\emptyset$ we
    have $(M,w)\not\models \neg \phi[p/\bot]$,
    while if $X\neq \emptyset$, then 
    (by finiteness of the domain)
    there must exist a minimal-size non-empty subset $X\subseteq W$ satisfying
    condition 2. Furthermore $X$ cannot be a singleton set, and therefore it must be the case that 
    $X=Y\cup Z$ for disjoint non-empty sets $Y$ and $Z$.
    It then follows that
       $(M[q_1\mapsto Y, q_2\mapsto Z],w)\not\models \phi[p/q_1\lor q_2] \rightarrow (\phi[p/q_1]\lor \phi[p/q_2])$.

    Having established already that $\phi$ is completely additive in $p$ in the finite, we will now argue that the same holds over arbitrary structures. In particular $\phi$ is monotone in $p$ in the finite. It follows from Proposition~\ref{prop:monotone-validity} that $\phi$ is monotone in $p$ over arbitrary structures. It therefore only remains to show that
    the second half of the definition of complete additivity is satisfied over arbitrary structures.
    Assume for the sake of a contradiction that this condition
    is \emph{not} satisfied. Let $M,w$ with $M=(W,(R_a)_{a\in A},V)$ be a witness, i.e., a pointed Kripke structure such that $M,w\models\phi$  and such that
    $M[p\mapsto \{v\}],w\not\models\phi$ for all $v\in V(p)$. We will show how to turn $M,w$ into a \emph{finite} counterexample to complete additivity.
    
    Let $M^{unr}=(W^{unr},(R_a)_{a\in A}^{unr},V^{unr})$ be the tree unraveling of the model $M$, defined in the usual way. 
    
    \medskip\par\noindent\emph{Claim 1:} $(M^{unr},\langle w\rangle)$ is also a witness to the failure of complete additivity.

    \medskip\par\noindent
        Indeed, 
    it follows from the fact that $(M,w)$ and $(M^{unr},\langle w\rangle)$ are bisimilar that 
    $(M^{unr},\langle w\rangle)\models\phi$. Furthermore, for any 
    $u\in V^{unr}(p)$, we have that $(M[p\to\{\pi(u)\}],w)$ is bisimilar
    to $(M^{unr}[p\to X_u],\langle w\rangle)$, where $\pi$ is the natural projection and where $X_u=\{u'\mid \pi(u)=\pi(u')\}$.
    By assumption, $(M[p\to\{\pi(u)\}],w)\not\models\phi$,
    and hence by bisimulation invariance, $(M^{unr}[p\to X_u],\langle w\rangle)\not\models\phi$. It follows by monotonicity (over 
    arbitrary Kripke structures) that $(M^{unr}[p\to\{u\}],\langle w\rangle)\not\models\phi$.

    Let $M'$ be the substructure of $M^{unr}$ containing only those 
    worlds that are reachable from $\langle w\rangle$ in at most
    $n$ steps, where $n$ is the modal depth of $\phi$. 
    Since the truth of $\phi$ in a structure depends only on this
    substructure, Claim 1 implies:

        \medskip\par\noindent\emph{Claim 2:} $(M',\langle w\rangle)$ is also a witness to the failure of complete additivity.

    \medskip\par\noindent
    Observe that $M'$ is a tree of bounded depth. With a small amount of 
    additional ``surgery'' it is possible to turn $M'$ into a finite Kripke structure 
    that is still a witness to the failure of complete additivity
    (this can be done by pruning the tree so that each node
    has at most two children of any given subtree isomorphism type, cf.~\cite{vBenthem1998:bisimulation}.
    However, as it turns out we can make life even easier by appealing  to a powerful
    result about the monadic second order logic (MSO) theory of unranked trees, namely the fact
    that \emph{the MSO-theory of the class of all  trees is equal
    to the MSO-theory of the class of all finitely branching trees}
    \cite{Walukiewicz2002}.%
    \footnote{This result follows from  Lemma 43 in \cite{Walukiewicz2002}, though it was not explicitly stated there.}
    Indeed, consider the MSO-sentence 
    \[ \Psi := \exists x(ST_x(\phi)\land \forall Q(\exists y(Py\land\forall z(Q(z)\leftrightarrow z=y))\to \neg ST_x(\phi)[P/Q]))\land depth_n \]
    where $depth_n$ is short for $\neg\bigvee_{a_1, \ldots, a_n\in A}\exists x_1\ldots x_{n+1}(R_{a_1}(x_1,x_2)\land\cdots\land R_{a_n}(x_n,x_{n+1}))$.
    By Claim 2, $\Psi$ is satisfied in a tree. It follows
    that $\Psi$ is also satisfied in a finitely branching tree.
    As every finitely branching tree of finite depth is a finite tree,
    this concludes our proof. 
    \end{proof}

It follows again that Theorem~\ref{thm:complete-additivity-pres} transfers to the finite.

\subsection{Continuity}

A modal formula $\phi$ is called
\emph{continuous} in a propositional variable $p$ if, for all pointed Kripke structures $(M,w)$, 
    it holds that $M,w\models\phi$ if and only if there is a finite set $X\subseteq V(p)$ such that $M[p\mapsto X],w\models\phi$.
Furthermore, $\phi$ is 
\emph{strongly continuous} in $p$ if
it is $n$-continuous in $p$ for some $n$,
which means that, for all pointed Kripke structures $(M,w)$, 
    it holds that $M,w\models\phi$ if and only if there is a set $X\subseteq V(p)$ of size at most $n$ such that $M[p\mapsto X],w\models\phi$. 
For example, it is easy to see that $[a] p$ is \emph{not} continuous in $p$ 
but $\langle a\rangle p$ \emph{is} (as is every other completely additive modal formula). 
This property of continuity has been identified and studied especially in the context of (modal and first-order) fixed point logics, since
    it guarantees stabilization of fixed-points at stage $\omega$ \cite{ExploringLogicalDynamics,Fontaine2008:continuous,Fontaine2018:some}.

The continuous fragment of modal logic has been characterized syntactically as follows:

\begin{thmC}[{\cite[Corollary 1]{Fontaine2008:continuous}}]
\label{thm:continuous}
For all modal formulas $\phi$ and propositional variables $p$, the following are equivalent:
   \begin{enumerate}
  \item
    $\phi$ is continuous in $p$, \vspace{1mm}
    \item $\phi$ it is equivalent to a formula in negation normal form such that no occurrence of $p$ is in the scope of a negation or of a $\Box$ operator.
    \end{enumerate}
\end{thmC}

Clearly, continuity is trivial as a property of formulas interpreted over finite structures.  Therefore, the above
characterization trivially does \emph{not} transfer to the finite. 
However, a variant of it does.

\begin{prop} \label{prop:strong-continuity}
   For all modal formulas $\phi$ and propositional variables $p$, the following statements are equivalent:
   \begin{enumerate}
       \item $\phi$ is continuous in $p$, \vspace{1mm}
       \item $\phi$ is strongly continuous in $p$ (i.e., $\phi$ is $n$-continuous in $p$ for some $n\in\mathbb{N}$).
   \end{enumerate}
\end{prop}

\begin{proof}
    The direction from 2 to 1 is trivial, while the direction from 1 to 2 follows from the compactness theorem for first-order logic,
    and holds in fact for every first-order formula. Indeed, let $\phi$ be any first-order formula containing a unary predicate $P$.
    Let $\Psi = \{\neg\psi_0, \neg\psi_1, \neg\psi_2, \ldots\}$
    with $\psi_n = \exists y_1, \ldots, y_n (P(y_1)\land\cdots\land P(y_n)\land\phi')$, with $\phi'$ obtained from $\phi$  by replacing every subformula of the form
    $P(z)$ by $z=y_1\lor\cdots\lor z=y_n$. By continuity, $\Phi$ entails $\neg\phi$. Therefore, by compactness, some finite subset
    $\{\neg\psi_{i_1}, \ldots, \neg\psi_{i_k}\}\subseteq \Psi$
    entails $\neg \phi$. This shows that $\phi$ is $n$-continuous for $n=\max\{i_1, \ldots, i_k\}$+1.   
\end{proof}

Next, we can apply our usual technique and correlate the fact that $\phi$ is $n$-continuous in $p$ to the validity of an associated formula $\phi'$:

\begin{prop} \label{prop:strong-continuity-finite}
   For all modal formulas $\phi$, propositional variables $p$, and $n\in\mathbb{N}$, the following are equivalent:
   \begin{enumerate}
       \item $\phi$ is $n$-continuous in $p$, \vspace{1mm}
       \item The following validity holds: $\models \phi[p/\bigvee_{i\in\{1\ldots n\}}(p_i)]\leftrightarrow \bigvee_{k\in\{1\ldots n\}} \phi[p/(\bigvee_{i\in\{1\ldots n\}\setminus\{k\}}(p_i)]$.
   \end{enumerate}
   where $p_1 \ldots, p_n$ are fresh propositional variables. This equivalence holds both over arbitrary structures and in the finite.
\end{prop}

\begin{proof} 
   From 1 to 2: Assume that $\phi$ is $n$-continuous in $p$. In particular,
   then, $\phi$ is monotone in $p$. It follows that the right-to-left direction of the bi-implication is valid. For the left-to-right direction, suppose that $M,w\models\phi[p/\bigvee_{i\in\{1,\ldots,n\}}(p_i)]$. Then, by
   $n$-continuity, there exists a set $X\subseteq \bigcup_{i\in\{1,\ldots,n\}} V(p_i)$ such that $|X|<n$ and $M[p\mapsto X],w\models\phi$. Next, for cardinality reasons, it must be the case that
   $X\subseteq\bigcup_{i\in\{1,\ldots,n\}\setminus \{k\}}V(p_i)$.
   Therefore, by the monotonicity property,
   $M,w\models\phi[p/\bigvee_{i\in\{1,\ldots,n\}\setminus\{k\}}(p_i)]$.\vspace{1mm}

   From 2 to 1, the exact same proof strategy applies as in the case of complete additivity (cf.~the proof of Proposition~\ref{prop:additivity-finite}). The main change is that the MSO-formula $\Psi$ used there to express failure of complete additivity is now replaced by the following
   MSO-formula expressing failure of $n$-continuity:
   \small{
       \[ \Psi' := \exists x(ST_x(\phi)\land \forall Q\Big(\forall y(Qy\to Py)\land \exists y_1\ldots y_n\forall z(Qz\to \bigvee_{i=1\ldots n}z=y_i)\to \neg ST_x(\phi)[P/Q]\Big)\land depth_n \]}
\end{proof}

It follows that, although Theorem~\ref{thm:continuous} does not transfer to the finite, it does transfer once continuity is replaced by strong continuity.

\subsection{Where we stand}

We have analyzed four examples of ubiquitous and useful semantic properties of modal formulas. In the simplest cases, a syntactic characterization transfers to the finite thanks to two features. One is that the preservation property can itself be expressed in modal terms, the other that modal logic has the finite model property.  In some more complex preservation cases, we found that some additional considerations were needed.
All this does not add up to a concrete necessary and sufficient criterion for transfer to the finite of properties of modal logic --- finding such a criterion (if one exists at all) would be interesting but is beyond our reach here. 

Our investigation also has a broader context. There are many decidable extensions of the modal fragment of first-order logic, such as the Guarded Fragment~\cite{ABN1998} and the Guarded Negation Fragment~\cite{Barany2015:guarded},
that retain the finite model property that we have used so extensively. We leave it as an open direction to  explore the transfer-to-the-finite behavior of semantic properties of formulas for such fragments. 

This concludes our modal-style exploration. In the next section, we give a similar example-based analysis of transfer to the finite for preservation properties in first-order logic whose motivation can be found in the model theory of modal logic.

\section{The first-order logic setting}\label{sec:FO}

The results in Section~\ref{sec:modal} were predominantly positive. In particular, we looked at a number of semantic properties of modal formulas, such as \emph{monotonicity} and \emph{preservation under induced substructures}, and we saw that  
the corresponding preservation theorems  do transfer to the finite. 
The picture is not as rosy for first-order logic. Indeed,  
it was shown in~\cite{Tait59} that 
the \L o\'{s}-Tarski theorem (characterizing the first-order formulas preserved under induced substructures) fails in the finite, and  it was shown
in \cite{AjtaiGurevich87} that, over finite structures, monotonicity can 
no longer be syntactically characterized in terms of positivity either. Nevertheless, as discussed in the introduction, there are examples of preservation theorems that do hold in the finite, 
such as Rossman's preservation theorem for (non-surjective) homomorphisms \cite{rossman2008homomorphism}.

In this section, we will consider several 
semantic properties of first-order formulas originating in modal correspondence theory, and we investigate whether their syntactic characterizations transfer to the finite. We start with a positive result about a major property of the central invariance of bisimulation. 

\subsection{Bisimulation invariance and bisimulation safety}

The \emph{Bisimulation Characterization Theorem}
and the \emph{Bisimulation Safety Theorem}, 
stated below, characterize the first-order formulas
that are invariant, respectively safe, for bisimulation. We will restrict here to first-order signatures consisting of
unary relations (corresponding to propositional variables) and binary relations (corresponding to modal operators).

Recall that a first-order formula $\phi(x)$ (over the appropriate signature) is \emph{bisimulation-invariant} if, whenever
$(M,w)$ and $(M',w')$ are bisimilar, then 
$M,w\models\phi(x)$ iff $M',w'\models\phi(x)$.

\begin{thmC}[Bisimulation Characterization Theorem \cite{benthemthesis}]
A first-order formula $\phi(x)$ is  
bisimulation-invariant if and only if it is logically equivalent to the
standard translation of some modal formula. 
\end{thmC}

Rosen~\cite{rosen-thesis}  showed that the Bisimulation Characterization Theorem holds also in the realm of finite models. That is,
a first-order formula $\phi(x)$ is  
bisimulation-invariant over finite structures if and only if it is equivalent, over finite structures, to the
standard translation of a modal formula.

    The Bisimulation Safety Theorem is another classic result, which syntactically characterizes the first-order  operations on binary relations that are safe for bisimulation.
A first-order formula $\phi(x,y)$ (over the appropriate signature) is \emph{safe for bisimulation} if every 
bisimulation between Kripke structures
$(M,w)$ and $(M',w')$ also satisfies the back and forth clauses w.r.t. the binary relation defined by $\phi$.
To state the  theorem, we first need
to define \emph{modal programs}. 
A modal program is an expression built 
by the following grammar: \vspace{-1.5ex}

\[ \pi ~~\text{:=}~~ p? \mid a \mid \pi;\pi' \mid \pi\cup\pi' \mid ~\sim\pi\]

\vspace{1mm}
where $p$ is a propositional variable 
and $a$ an atomic program.
Each modal program, interpreted in a 
Kripke structure $M=(W,(R_a)_{a\in A},V)$, defines a binary
relation: \vspace{-0.5ex}

\newcommand{\sem}[1]{[\![#1]\!]}

\[\begin{array}{@{}lll}
\sem{p?}_M = \{ \langle w,w\rangle \mid w\in V(p)\} \\ 
\sem{a}_M = R_a \\
\sem{\pi;\pi'}_M = \sem{\pi}_M \cdot \sem{\pi'}_M \\
\sem{\pi\cup\pi'}_M = \sem{\pi}_M \cup \sem{\pi'}_M \\
\sem{\sim\pi}_M = \{ \langle w,w\rangle\mid \text{$w\in W$ and there is no $v\in W$ s.t.~$\langle w,v\rangle\in\sem{\pi}_M$}\}
\end{array}\]
\vspace{0.5ex}

The standard translation from modal formulas to first-order formulas has a natural analogue for modal programs, yielding first-order formulas in two free variables:

\vspace{-1ex}

\[\begin{array}{llll}
ST_{x,y}(a) &=& R_a(x,y) & \text{for $a\in A$}\\
ST_{x,y}(p?) &=& (x=y) \land P(x) & \text{for $p$ a propositional variable} \\
ST_{x,y}(\pi;\pi') &=& \exists z (ST_{x,z}(\pi)\land ST_{z,y}(\pi'))\\
ST_{x,y}(\pi\cup\pi') &=& ST_{x,y}(\pi)\lor ST_{x,y}(\pi') \\
ST_{x,y}(\sim\pi) &=& (x=y)\land \neg\exists z (ST_{x,z}(\pi))
\end{array}\]

\vspace{0.5ex}

\begin{thmC}[Bisimulation Safety Theorem \cite{vBenthem1998:bisimulation}]
\label{thm:safety}
A first-order formula $\phi(x,y)$ is safe for 
bisimulation if and only if it is equivalent to the
standard translation of a modal program. 
\end{thmC}

 The bisimulation safety theorem can also be cast as a semantic characterization for a syntactic fragment of Tarski's relation algebra, cf.~\cite{Bogaerts2024:preservation}.
 We will show that Theorem~\ref{thm:safety} transfers to the finite as well. 
 
\begin{thmC}
\label{thm:safety-finite}
Over finite structures, a first-order formula $\phi(x,y)$ is safe for 
bisimulation if and only if it is equivalent to the
standard translation of a modal program. 
\end{thmC}

 The core ingredient of the proof is Proposition~\ref{prop:additivity-finite} above, which (in combination with the finite model property of modal logic) tells us that a modal formula
    is completely additive in the finite if and only if
    it is completely additive over arbitrary Kripke structures.

By a \emph{$p$-free program} we will mean a program expression that does not contain any occurrence of $p$.
The following lemma exhibits a relationship between modal programs and completely additive formulas. 

\begin{lem}[\cite{vBenthem1998:bisimulation}]\label{lem:formulas-to-programs}
    For each modal formula $\theta$
    of the form as described in
    Theorem~\ref{thm:complete-additivity-pres},
    there is a $p$-free modal program $\pi_\theta$ 
    such that,
    for all Kripke structures $M$,
    it holds that 
    $M,w\models\theta$ iff 
    there is a pair $\langle w,v\rangle\in \sem{\pi_\theta}_M$ with $M,v\models p$.
\end{lem}

Van Benthem's proof of Theorem~\ref{thm:safety} in \cite{vBenthem1998:bisimulation}  uses the Bisimulation Characterization Theorem as a black box. It  goes as follows: suppose $\phi(x,y)$ defines a bisimulation safe
operation on binary relations.
Let $P$ be a fresh unary predicate. Then  it follows that
$\psi(x) := \exists y(\phi(x,y)\land P(y))$ is bisimulation-invariant. Therefore, by the Bisimulation Characterization Theorem,
$\psi(x)$ is equivalent to a 
modal formula $\chi$ (which may use a propositional variable
$p$ corresponding to the unary predicate $P$). 
It follows that $\chi$ is completely additive in 
$p$ (because, clearly, $\psi$ is). Therefore, 
by Proposition~\ref{prop:additivity-finite} together with Lemma~\ref{lem:formulas-to-programs}, using some straightforward syntactic manipulation, 
$\psi$ can be written in the form 
$\langle \pi\rangle p$ for some $p$-free modal program $\pi$. 
It then follows that $\pi$ defines the same relation as $\phi(x,y)$.

\vspace{1ex}

Theorem~\ref{thm:safety-finite}
follows using same argument, but
using Rosen's theorem \cite{rosen-thesis} instead of the Bisimulation Preservation Theorem, and using Proposition~\ref{prop:additivity-finite}.
Suppose $\phi(x,y)$ defines a bisimulation safe
operation on binary relations, in the finite.
Let $P$ be a fresh unary predicate. Then
 (by the same arguments as in~\cite{vBenthem1998:bisimulation}),
$\psi(x) := \exists y(\phi(x,y)\land P(y))$ is bisimulation-invariant in the finite. Therefore, by Rosen's theorem,
$\psi(x)$ is equivalent, in the finite, to a 
modal formula $\chi$ (which may use a propositional variable
$p$ corresponding to the unary predicate $P$). 
It follows that $\chi$ is completely additive in 
$p$ in the finite (because, clearly, $\psi$ is). Therefore, 
by Proposition~\ref{prop:additivity-finite} in combination with the finite model property, $\chi$ is complete additive in $p$ over all Kripke structures. Hence, 
by Theorem~\ref{thm:complete-additivity-pres} together with Lemma~\ref{lem:formulas-to-programs},
$\psi$ can be written in the form 
$\langle \pi\rangle p$ for some $p$-free modal program $\pi$. 
It then follows that $\pi$ defines the same relation as $\phi(x,y)$, in the finite. This concludes the proof.

\vspace{-2ex}

\subsection{Basic modal preservation operations on frames}

Our second set of examples is concerned with some basic semantic properties suggested by the behavior of modal formulas on \emph{frames} rather than models (as in the preceding result).  Recall that a modal formula $\phi$ is said to be 
\emph{valid} on a Kripke frame $F$
if $\phi$ is globally satisfied in 
every Kripke model whose underlying
frame is $F$. For example, the modal
formula $\Box p\to\Box\Box p$ is
valid on a frame precisely if it is 
transitive.
A 
class $K$ of frames is said to be
\emph{modally definable} if there is a set  $\Phi$ of modal formulas such that
$K$ consists precisely of those frames
on which each modal formula in $\Phi$ is valid.

There are four well-known preservation properties of modal formulas in this setting, which occur, for instance, in the well-known  Goldblatt-Thomason Theorem that we will discuss in Section~\ref{sec:GbTh} below.

\begin{defi}[Generated subframe]
A frame $F=(W,(R_a)_{a\in A})$ is a \emph{generated subframe}
of a frame $G=(W',(R'_a)_{a\in A})$ if
$F$ is an induced substructure of $G$ and 
for all $(w,v)\in R_a$ ($a\in A$), 
if $w\in W$ then $v\in W$.
\end{defi}

\begin{defi}[Disjoint union]
Let $F_i=(W_i,(R^i_a)_{a\in A})$ (for $i\in I$) be 
a collection of frames  with disjoint domains. The \emph{disjoint union}
$\uplus_{i\in I} F_i$ is the frame 
$(W',(R'_a)_{a\in A})$ where 
$W'=\bigcup_{i\in I} W_i$ and $R'_a=\bigcup_{i\in I} R^i_a$.
\end{defi}

\begin{defi}[Bounded morphic image]
A bounded morphism from a frame $F=(W,(R_a)_{a\in A})$ to a frame $G=(W',(R'_a)_{a\in A})$  is a map $f:W\to W'$
satisfying, for all $i\leq n$:
\begin{enumerate}
    \item $(w,v)\in R_a$ implies $f(w),f(v))\in R'_a$, and\vspace{1mm}
    \item $(f(w),u)\in R'_a$ implies that $(w,v)\in R_a$ for some $v\in f^{-1}(u)$.
\end{enumerate}
If there is a surjective bounded morphism from $F$ to $G$, we say that $G$ is a
bounded morphic image of $F$.
\end{defi}

\begin{defi}[Ultrafilter extension] 
Given a frame $F = (W,(R_a)_{a\in A})$, the ultrafilter extension of $F$, denoted by $ueF$, is the frame
$(Uf(W),(R^{ue}_a)_{a\in A})$, with $Uf(W)$ the set of ultrafilters over $W$ and  $R^{ue}_a= \{(u,v)\mid 
\text{ for all $X \in v$, $\{w \in W \mid (w,v)\in R_a \text{ for some } v\in X\}\in u$}\}$.
\end{defi}

In what follows in this section, we will restrict attention
to first-order sentences over a signature
consisting of binary relations only. 
Such sentences describe properties of 
(multi-modal) Kripke frames rather than modal models carrying a valuation. 
\subsection{Generated subframes}

A FO-formula is said to be \emph{bounded}
if it is built up from literals (including equality literals) using
the conjunction, disjunction, and \emph{bounded
quantification} of the form 
$\exists x(Ryx\land\cdots)$ and/or
$\forall x(Ryx\to\cdots)$, where $R$
is a binary relation and $y$ is a variable 
distinct from $x$. It is easy to see that,
for bounded FO formulas $\phi(x)$, whether
the formula is true of an element $w$ in 
a structure $M$ depends only on the submodel
of $M$ generated from $w$ along the binary relations.  A FO-formula is 
\emph{$\exists$-bounded} if it is similarly built up, but in addition, we allow also unbounded universal quantification. 

\begin{thmC}[\cite{FefermanKreisel1966}]
\label{thm:generated}
    A first-order sentence $\phi$ is preserved under  generated subframes if and only if 
    $\phi$ is equivalent to an $\exists$-bounded FO sentence.
\end{thmC}

As we will show, this preservation theorem fails in the finite. 
\begin{thm}\label{thm:generated-fail}
    There is a FO sentence that is preserved under taking generated subframes in the finite but is not equivalent in the finite to any $\exists$-bounded FO  sentence.
\end{thm}
Our proof of this theorem proceeds in two steps.
First, we define a first-order sentence $\chi$ that is preserved under generated subframes in the finite.
Second, for every $n$, we construct finite structures $\mathfrak{A}_n$ and $\mathfrak{B}_n$ such that $\mathfrak{A}_n \models \chi$ while $\mathfrak{B}_n \not\models \chi$.
We then show that no $\exists$-bounded first-order sentence of quantifier rank at most $n$ can distinguish between $\mathfrak{A}_n$ and $\mathfrak{B}_n$, witnessing that $\chi$ is not definable by any $\exists$-bounded sentence of quantifier rank $n$.

It was recently shown in \cite{Bogaerts24:preservation} that there is a first-order formula $\psi(x)$ that is invariant for generated submodels in the finite but that is not equivalent in the finite to a formula in the bounded fragment.
The formula in question is
$\psi(x) := \exists u(R_1(u,x)\land\phi(u))$ where
$\phi(u)$ is the conjunction of the following five FO-formulas:
  \begin{itemize}
    \item $R_3(u,u)$
    \item $\forall v(R_3(u,v)\to \exists w(R_3(v,w)\land R_2(w,v)))$
    \item $\forall vw(R_3(u,v)\land R_3(v,w)\to R_3(u,w))$
    \item $\neg\exists v R_2(u,v)$
    \item $\forall vw(R_3(u,v)\land \exists^{\geq 2}s (R_2(v,s)\land R_3(u,s))\land R_1(u,w)\to R_4(w,v))$ \footnote{Here we slightly modify the original formula; this modification does
not affect the result in \cite{Bogaerts24:preservation}.}
\end{itemize}

Intuitively, the formula $\psi(x)$ forces the existence of an $R_1$-predecessor $u$ of $x$ and an infinite $R_2$-chain arriving at $u$, which, in a finite structure, must eventually loop, triggering a forward $R_4$-edge from $x$ that eventually makes the entire chain belong to the generated submodel.
% See Figure~\ref{fig:counterexample-forward}. 
We now turn this observation into a sentence-level counterexample.
Define
\[
\chi := \forall x\bigl(\exists y\, R_4(x,y) \rightarrow \psi(x)\bigr).
\]
Since $\psi(x)$ is invariant under generated submodels in the finite, we
get the following result.

\begin{lem}\label{lem:chi-preserved}
The first-order sentence $\chi$ is preserved under taking generated subframes over finite frames.
\end{lem}

We now proceed to the second step of the proof, namely the construction of the structures $\mathfrak{A}_n$ and $\mathfrak{B}_n$.
Intuitively, the structure $\mathfrak{A}_n$ is designed so that the unique $R_1$-edge in $\mathfrak{A}_n$ is not reachable within depth $n$ from the root of the tree.
This will ensure that $\exists$-bounded formulas of quantifier rank at most $n$ cannot detect its presence.

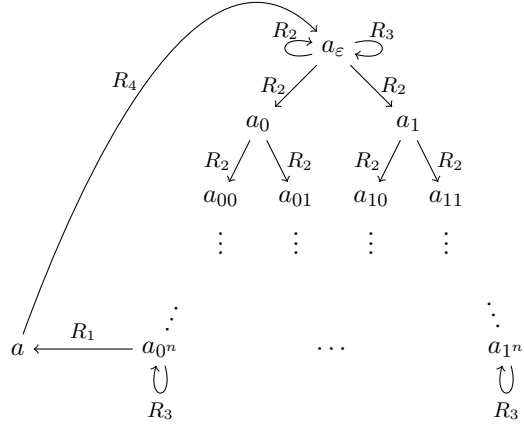
\begin{figure}
\centering
\begin{tikzpicture}
  [align=center,scale=1,
   every node/.style={scale=.9},
   ell/.style={inner sep=1pt, font=\large},
   edge label/.style={font=\footnotesize}]

  % nodes 
    \node (a)  at (-1.2,0) {$a$};
    \node (be) at (3,4) {$a_{\varepsilon}$};
    \node (b0) at (2,3) {$a_{0}$};
    \node (b1) at (4,3) {$a_{1}$};
    \node (b00) at (1.5,2) {$a_{00}$};
    \node (b01) at (2.5,2) {$a_{01}$};
    \node (b10) at (3.5,2) {$a_{10}$};
    \node (b11) at (4.5,2) {$a_{11}$};
    % ellipses below leaves (full binary tree continues)
    \node[ell] (e00) at (1.5,1.5) {$\vdots$};
    \node[ell] (e01) at (2.5,1.5) {$\vdots$};
    \node[ell] (e10) at (3.5,1.5) {$\vdots$};
    \node[ell] (e11) at (4.5,1.5) {$\vdots$};
    % deep chain endpoints
    \node (b0n) at (0.7,0) {$a_{0^n}$};
    \node (b1n) at (5.3,0) {$a_{1^n}$};

    % sloped ellipses (invisible paths)
    \draw[draw=none] (b0n) -- node[pos=0.13,sloped] {$\cdots$} (b00);
    \draw[draw=none] (b1n) -- node[pos=0.13,sloped] {$\cdots$} (b11);
    \draw[draw=none] (b0n) -- node[midway] {$\cdots$} (b1n);

    % R1 edge
    \draw [-to] (b0n) to node[above, edge label] {$R_1$} (a);
    
    % R2 edges (tree)
    \draw [-to] (be) -- node[left,edge label] {$R_2$} (b0);
    \draw [-to] (be) -- node[right,edge label] {$R_2$} (b1);
    \draw [-to] (b0) -- node[left,edge label] {$R_2$} (b00);
    \draw [-to] (b0) -- node[right,edge label] {$R_2$} (b01);
    \draw [-to] (b1) -- node[left,edge label] {$R_2$} (b10);
    \draw [-to] (b1) -- node[right,edge label] {$R_2$} (b11);
    % R2 loop at root
    \draw [-to] (be) to [loop left] node[above,edge label] {$R_2$} (be);

    %R3 loops
    \draw [-to] (be) to [loop right] node[above,edge label] {$R_3$} (be);
    \draw [-to] (b0n) to [loop below] node[below,edge label] {$R_3$} (b0n);
    \draw [-to] (b1n) to [loop below] node[below,edge label] {$R_3$} (b1n);

    % R4 edge
    \draw [-to] (a) to [out=70] node[left,edge label] {$R_4$} (be);

\end{tikzpicture}

\vspace{1ex}
\caption{The structure $\mathfrak{A}_n$. In addition to the edges shown, there is an $R_3$-edge from $a_w$ to $a_u$ whenever $u$ is a proper prefix of $w$.}
\label{An}
\end{figure}

Fix $n \ge 1$.
The structure $\mathfrak{A}_n=(A_n,R_1,R_2,R_3,R_4)$ is defined as follows.
The universe $A_n$ consists of a distinguished element $a$ and elements $a_w$ indexed by binary words $w$ with $|w|\le n$, forming a full binary tree of height $n$ with root $a_\varepsilon$ (see Figure~\ref{An}).
The relations are given below.

\begin{itemize}
  \item[$R_1$:]
  A single edge $R_1(a_{0^n},a)$ from the leftmost leaf to the distinguished
  element. This is the only edge connecting the tree to $a$ that is absent
  in the variant $\mathfrak{A}_n'$ defined below.

  \item[$R_2$:]
  Tree edges encoding the binary branching structure: for each $w$ with
  $|w|<n$, the node $a_w$ has exactly two $R_2$-successors, namely
  $a_{w0}$ and $a_{w1}$. In addition, the root $a_\varepsilon$ is
  $R_2$-reflexive.

  \item[$R_3$:]
  The ancestor relation on the tree: whenever $u$ is a proper prefix of $w$, there is an $R_3$-edge from $a_w$ to $a_u$. Moreover, the root $a_\varepsilon$ and all leaves $a_w$ with $|w|=n$ are $R_3$-reflexive. For readability, only some $R_3$-edges are shown in the figure.

  \item[$R_4$:]
  The edge linking the distinguished element to the tree, namely $R_4(a,a_\varepsilon)$.

  % together with a loop $R_4(a,a)$.
\end{itemize}

\noindent
No other $R_i$-edges are present.

\medskip

Let $\mathfrak{A}_n'$ be the structure obtained from $\mathfrak{A}_n$ by
removing the $R_1$-edge $R_1(a_{0^n},a)$.
Let $\mathfrak{B}_n$ be a copy of $\mathfrak{A}_n'$.
For notational convenience, we rename the distinguished element $a$ as $b$ and each element $a_w$ as $b_w$. Then $\mathfrak{A}_n \models \chi$, whereas $\mathfrak{B}_n \not\models \chi$.

To show that no $\exists$-bounded FO sentence of quantifier rank $n$ can distinguish between $\mathfrak{A}_n$ and $\mathfrak{B}_n$ (Lemma~\ref{Lem:An-Bn}), we introduce the $\exists$-bounded game characterizing $\exists$-bounded FO sentences.

The $\exists$-bounded game is a variant of the
Ehrenfeucht--Fra\"{\i}ss\'e game.
For every natural number $n$, the \emph{$n$-round $\exists$-bounded game} played on a pair of structures $\mathfrak{A}$ and $\mathfrak{B}$ of same signature proceeds as follows.
After $i$ rounds, the current configuration is the pair
$
(\vec{x},\vec{y})=(x_1,\dots,x_i \ ;\ y_1,\dots,y_i),
$
where $x_j$ and $y_j$ are the elements selected in round~$j$ in $\mathfrak{A}$ and $\mathfrak{B}$, respectively.
In round $i+1$, Spoiler chooses one of the two structures.
If Spoiler plays in $\mathfrak{B}$, he may select an arbitrary element $y_{i+1}\in B$.
If Spoiler plays in $\mathfrak{A}$, he must select an element $x_{i+1}\in A$ that is \emph{bounded} by the current configuration, that is, $x_{i+1}$ must be reachable from some previously chosen element $x_j$ by one step along the underlying relations.
Duplicator then responds by selecting an arbitrary element in the other structure, thereby extending the configuration to
$
(x_1,\dots,x_{i+1} \ ;\ y_1,\dots,y_{i+1}).
$

A configuration $(\vec{x},\vec{y})$ forms a \emph{partial isomorphism} if the mapping $x_j\mapsto y_j$ is an isomorphism between the substructures of $\mathfrak{A}$ and $\mathfrak{B}$ induced by $\vec{x}$ and $\vec{y}$, respectively.
Duplicator wins the $n$-round $\exists$-bounded game if she can maintain a partial isomorphism after every round; otherwise, Spoiler wins.

It can be shown that Duplicator has a winning strategy in the $n$-round $\exists$-bounded game on $\mathfrak{A}$ and $\mathfrak{B}$ if and only if, for every $\exists$-bounded first-order sentence $\varphi$ of quantifier rank at most $n$, $\mathfrak{A} \models \varphi$ implies $\mathfrak{B} \models \varphi$. The proof follows standard arguments and is omitted.

\begin{lem}\label{Lem:An-Bn}
For every $\exists$-bounded sentence $\theta$ of quantifier rank at most $n$, if $\mathfrak{A}_n \models \theta$ then $\mathfrak{B}_n \models \theta$.
\end{lem}

\begin{proof}
It suffices to show that Duplicator has a winning strategy in the $n$-round $\exists$-bounded game on $\mathfrak{A}_n$ and $\mathfrak{B}_n$.

We introduce some notation.
We write $\mathfrak{T}_{a_w}$ and $\mathfrak{T}_{b_w}$ for the subtrees rooted at $a_w$ and $b_w$ in $\mathfrak{A}_n$ and $\mathfrak{B}_n$, respectively; that is,
the induced substructures on the sets
\[
\{\, a_u \mid w \text{ is a prefix of } u \,\}
\quad\text{and}\quad
\{\, b_u \mid w \text{ is a prefix of } u \,\},
\]
respectively. 
For simplicity, we let the initial configuration be $(a\ ;\ b)$.
We describe a winning strategy for Duplicator in the $n$-round $\exists$-bounded game between $\mathfrak{A}_n$ and $\mathfrak{B}_n$ starting from this configuration. During the game, we denote by $x_i$ and $y_i$ the elements selected in round $i$ in $\mathfrak{A}_n$ and $\mathfrak{B}_n$, respectively.
Thus, the configuration after round $i$ is written as
\[
(a,x_1,\dots,x_i \ ;\ b,y_1,\dots,y_i).
\]

\noindent
\textbf{Claim.} Duplicator has a strategy ensuring that, after each round $i\le n$, the configuration $(a,x_1,\dots,x_i \; ; \; b,y_1,\dots,y_i)$ satisfies the following invariants (we may assume that $a\notin\{x_1,\dots,x_i\}$ and
$b\notin\{y_1,\dots,y_i\}$):
\begin{enumerate}
  \item[\textup{(I)}]
  There exists an isomorphism
  \[
  f:\mathfrak{T}_{a_{\varepsilon}} \to \mathfrak{T}_{b_{\varepsilon}}
  \]
  that extends $(x_1,\dots,x_i \ ;\ y_1,\dots,y_i)$, i.e., $f(x_j)=y_j$ for all $j\le i$.
  \item[\textup{(II)}]
  Each $x\in\{x_1,\dots,x_i\}$ lies in $\mathfrak{T}_{a_{\varepsilon}}\setminus \mathfrak{T}_{a_{0^{i}}}$.
\end{enumerate}
Invariant~\textup{(I)} ensures that $(x_1,\dots,x_i;y_1,\dots,y_i)$ forms a partial isomorphism between $\mathfrak{A}_n$ and $\mathfrak{B}_n$.
Moreover, invariant~\textup{(II)} implies that no $x\in\{x_1,\dots,x_i\}$ is equal to $a_{0^n}$, so the $R_1$-edge from $a_{0^n}$ to $a$ is never involved in the game.
Together, these two invariants ensure that $(a,x_1,\dots,x_i \ ;\ b,y_1,\dots,y_i)$ forms a partial isomorphism.
Thus, Duplicator can win the game by maintaining these invariants.

\smallskip
\noindent
\emph{Proof of the claim.}
We proceed by induction on $i$.

\smallskip
\noindent
\emph{Base case ($i=0$).}
Trivial.

\smallskip
\noindent
\emph{Induction step.}
Assume the invariants hold after round $i$ and that the configuration is $(a,x_1,\dots,x_i \ ;\ b,y_1,\dots,y_i)$.
We consider round $i+1$.

\smallskip
\noindent
\emph{Case~1: Spoiler plays in $\mathfrak{A}_n$.}
Suppose Spoiler selects an element $x_{i+1}$.
By invariant~\textup{(I)}, there exists an isomorphism
$f:\mathfrak{T}_{a_{\varepsilon}} \to \mathfrak{T}_{b_{\varepsilon}}$ extending the current configuration.
Duplicator responds by choosing $y_{i+1}=f(x_{i+1})$. Invariant~\textup{(I)} is clearly maintained.

To see that invariant~\textup{(II)} is also maintained, note that by the induction hypothesis, $x_j\in\mathfrak{T}_{a_{\varepsilon}}\setminus \mathfrak{T}_{a_{0^{i}}}$ for all $j\le i$.
Since Spoiler is restricted to the set $\{a,x_1,\dots,x_i\}$ in $\mathfrak{A}_n$, it follows that $x_{i+1}\in\mathfrak{T}_{a_{\varepsilon}}\setminus \mathfrak{T}_{a_{0^{i+1}}}$.
Hence invariant~\textup{(II)} is maintained for round $i+1$.

\smallskip
\noindent
\emph{Case~2: Spoiler plays in $\mathfrak{B}_n$.}
Suppose Spoiler selects an element $y_{i+1}$, and let $f:\mathfrak{T}_{a_{\varepsilon}} \to \mathfrak{T}_{b_{\varepsilon}}$ be an isomorphism as in invariant~\textup{(I)}.

\smallskip
\noindent
\emph{Subcase~2a.}
If $f^{-1}(y_{i+1})\in\mathfrak{T}_{a_{\varepsilon}}\setminus \mathfrak{T}_{a_{0^{i+1}}}$, Duplicator responds with
$x_{i+1}=f^{-1}(y_{i+1})$.
Invariant~\textup{(I)} is clearly maintained.
Since $\mathfrak{T}_{a_{0^{i+1}}}\subseteq \mathfrak{T}_{a_{0^{i}}}$,
we have
$x_{i+1}\in\mathfrak{T}_{a_{\varepsilon}}\setminus \mathfrak{T}_{a_{0^{i+1}}}$, so invariant~\textup{(II)} is also maintained.

\smallskip
\noindent
\emph{Subcase~2b.}
If $f^{-1}(y_{i+1})\in \mathfrak{T}_{a_{0^{i+1}}}$, Duplicator responds with an element $x_{i+1}$ symmetric to $f^{-1}(y_{i+1})$ with respect to $a_{0^{i}}$.
More precisely, if $f^{-1}(y_{i+1})=a_{0^{i+1}\cdot u}$ for some binary word $u$, then Duplicator selects $x_{i+1}=a_{0^{i}\cdot 1\cdot u}$.
Invariant~\textup{(II)} is maintained, since $a_{0^{i}\cdot 1\cdot u}\in \mathfrak{T}_{a_{\varepsilon}}\setminus \mathfrak{T}_{a_{0^{i+1}}}$.
To see that invariant~\textup{(I)} is also maintained, let $h$ be the automorphism of $\mathfrak{T}_{a_{\varepsilon}}$ defined as follows: for every $a_w$, if $w=0^{i+1}\cdot u$ for some word $u$, we set
$h(a_w) = a_{0^{i}\cdot 1\cdot u}$ and $h(a_{0^{i}\cdot 1\cdot u}) = a_{0^{i+1}\cdot u}$. If $w$ does not have $0^{i+1}$ as a prefix, we let $h(a_w) = a_w$.
Define $f' := f \circ h$. Note that $f'(x_{i+1})=y_{i+1}$.
It remains to verify that $f'$ is an isomorphism between $\mathfrak{T}_{a_{\varepsilon}}$ and $\mathfrak{T}_{b_{\varepsilon}}$ extending $(x_1,\dots,x_{i+1}; y_1,\dots,y_{i+1})$.
Since $f$ is an isomorphism and $h$ is an automorphism of $\mathfrak{T}_{a_{\varepsilon}}$, their composition $f'$ is again an isomorphism. Moreover, by the induction hypothesis, every element $x \in \{x_1,\dots,x_i\}$ lies in $\mathfrak{T}_{a_{\varepsilon}} \setminus \mathfrak{T}_{a_{0^{i}}}$.
Hence $h(x_j) = x_j$ for all $j \le i$, and therefore
$f'(x_j) = f(x_j) = y_j$ for all $j \le i$.
Finally, we have $x_{i+1} = a_{0^{i}\cdot 1\cdot u}$ for some $u$,
and thus $x_{i+1} \in
\mathfrak{T}_{a_{\varepsilon}} \setminus \mathfrak{T}_{a_{0^{i+1}}}$.
Consequently, $f'(x_{i+1}) = f(h(x_{i+1})) = f(a_{0^{i+1}\cdot u})= y_{i+1}$, and invariant~\textup{(I)} is maintained.

This completes the induction and hence the proof of the claim. Since Duplicator can maintain the invariants for $n$ rounds, she wins the $n$-round $\exists$-bounded game. The statement of the lemma follows.
\end{proof}

\medskip
\noindent
\emph{Proof of Theorem~\ref{thm:generated-fail}.}
By Lemma~\ref{lem:chi-preserved}, the sentence $\chi$ is preserved under taking generated subframes over finite frames.
Furthermore, for every $n$, there exist finite structures $\mathfrak A_n$ and $\mathfrak B_n$ such that $\mathfrak A_n \models \chi$ while $\mathfrak B_n \not\models \chi$.
By Lemma~\ref{Lem:An-Bn}, every $\exists$-bounded first-order sentence of quantifier rank at most $n$ that is true in $\mathfrak A_n$ is also true in $\mathfrak B_n$.
Since $n$ is arbitrary, it follows that $\chi$ is not equivalent in the finite to any $\exists$-bounded first-order sentence. This completes the proof of the theorem. \hfill $\square$

\medskip
We leave it as an open question whether the above counterexample can be modified to use a single binary relation.

\subsection{Disjoint unions}\label{subsec:disjoint unions}
Recall that a first-order sentence $\phi$ is preserved under disjoint unions if for every family of structures $(\mathfrak A_i)_{i\in I}$,
if $\mathfrak A_i\models\phi$ for all $i\in I$, then
$\biguplus_{i\in I}\mathfrak A_i\models\phi$.
In \cite{vanBenthem1985}, it was conjectured 
that every first-order sentence preserved under disjoint unions is equivalent to a sentence of the form
$\forall x\alpha(x)$, where $\alpha(x)$ is constructed from atomic formulas
and their negations using $\land$, $\lor$, $\exists$, and restricted universal quantifiers of the forms $\forall y(Rzy \to \cdots)$ and
$\forall y(Ryz \to \cdots)$, where $z$ is a variable distinct from $y$.

We provide a simple counterexample showing that this conjecture is false. It works both over arbitrary frames and over finite frames.
% For simplicity, we formulate the counterexample over vocabularies containing two unary predicate symbols $P$ and $Q$. 
% The example easily adapts to vocabularies containing only binary relation symbols.
Consider the following first-order sentence with two binary predicate symbols:
\[
\psi := \exists x P (x,x)  \wedge
\bigl((\exists^{\le 1} x  Q(x,x) \vee (\exists^{\ge 2} x P(x,x)\bigr).
\]
Thus, $\psi$ states that there exists a $P$-reflexive element, and that either there is at most one $Q$-reflexive element
or at least two $P$-reflexive elements.
It is immediate that $\psi$ is preserved under disjoint unions, since the disjoint union of two models of $\psi$ contains at least two $P$-reflexive elements.

Next, consider two structures $\mathfrak A$ and $\mathfrak B$. 
Structure $\mathfrak A$ has domain $\{a,b\}$ with $P(a,a)$ and $Q(b,b)$, 
while structure $\mathfrak B$ has domain $\{c,d,e\}$ with $P(c,c)$ and $Q(d,d), Q(e,e)$.
Clearly, $\mathfrak A \models \psi$ while $\mathfrak B \not\models \psi$.
However, for every first-order sentence of the form
$\forall x\alpha(x)$, with all universal quantifiers in
$\alpha(x)$ restricted, if $\mathfrak A \models \forall x\,\alpha(x)$, then
$\mathfrak B \models \forall x\,\alpha(x)$.  
This can be shown using an Ehrenfeucht--Fra\"iss\'e game corresponding to the formula class above on the pair $(\mathfrak A,\mathfrak B)$.
In this game, Spoiler's moves other than the initial one are restricted to either placing a pebble in the left structure or placing a pebble on a node in the right structure that is reachable from some already pebbled node.
Thus Spoiler cannot force a configuration in which both $Q$-reflexive nodes $d$ and $e$ of $\mathfrak B$ are pebbled simultaneously. It follows that Duplicator wins the game on $(\mathfrak A,\mathfrak B)$.

\begin{rem}\label{rem:disjoint-unions}
Observe that on the class of structures in which exactly one element is $P$-reflexive, the sentence $\psi$ above is equivalent to
$\exists^{\le 1}x Q(x,x)$. 
More generally, if $\phi$ is any first-order sentence and we define
$
\psi_\phi := \exists x P(x,x)  \wedge  (\phi  \vee  \exists^{\ge 2}x\,P(x,x)),
$ then $\psi_\phi$ is also preserved under disjoint unions, for the same reason as for $\psi$ above.
Moreover, over the class of structures with exactly one $P$-reflexive element,
$\psi_\phi$ is equivalent to $\phi$. Thus, over this class of structures,
sentences preserved under disjoint unions have the full expressive power of first-order logic.
This may indicate that preservation under disjoint unions is difficult to capture by a ``natural'' syntactic characterization.

Nevertheless, our counterexample relies on the presence of identity in the language. We leave it as an open problem whether the conjecture holds for
first-order logic without identity.
\end{rem}

While preservation under disjoint unions does not admit the conjectured syntactic characterization, the situation for \emph{invariance} is different.
Already in the classical setting, \cite{vanBenthem1985} shows that a first-order sentence is invariant under disjoint unions if and only if it is equivalent to a sentence of the form $\forall x\alpha(x)$, where $\alpha(x)$ is a two-way restricted first-order formula, that is, a formula in which the only allowed quantifiers are restricted universal quantifiers $\forall y (Rzy \to \cdots)$ and $\forall y (Ryz \to \cdots)$, and existential quantifiers $\exists y (Ryz \wedge \cdots)$ and $\exists y (Rzy \wedge \cdots)$, where $z$ is a variable distinct from $y$.
We show,
building on results from~\cite{OTTO2004173}, that this invariance theorem also holds in the finite.

We first recall some standard notions of locality. These notions will also be used in Section~\ref{subsec:McKinsey}.
The \emph{Gaifman graph} $\mathcal{G}(\mathfrak A)$ of a structure $\mathfrak A$ is the simple undirected graph with vertex set $A$,
where two distinct elements $u,v\in A$ are adjacent if they occur together in some tuple of a relation of $\mathfrak A$.
The \emph{distance} $\dist_{\mathfrak A}(u,v)$ between two elements $u,v\in A$ is their distance in the Gaifman graph $\mathcal G(\mathfrak A)$.
For $r\in\mathbb N$ and $a\in A$, the \emph{$r$-neighbourhood} $\mathcal N_r^{\mathfrak A}(a)$ of $a$ in $\mathfrak A$ is the substructure of $\mathfrak A$ induced by the set
$
\{v\in A : \dist_{\mathfrak A}(a,v)\le r\}.$

\begin{prop}
For a first-order sentence $\theta$, the following are equivalent over finite structures:
\begin{enumerate}
\item $\theta$ is invariant under disjoint unions, i.e., for every family of structures $(\mathfrak A_i)_{i\in I}$,
$
\mathfrak A_i \models \theta \text{ for all } i\in I
\Longleftrightarrow
\biguplus_{i\in I}\mathfrak A_i \models \theta.
$
  \item $\theta$ is equivalent to a sentence of the form $\forall x\,\varphi(x)$, where $\varphi(x)$ is a two-way restricted formula.
\end{enumerate}
\end{prop}

\begin{proof}
We only show that (1) implies (2); the converse is straightforward.

Recall that A first-order formula $\alpha(x)$ is local if there exists a $r\in\mathbb{N}$ such that for every structure $\mathfrak A$ and $a\in A$, whether $\mathfrak A\models \alpha(a)$ holds depends only on the $r$-neighbourhood of $a$ in $\mathfrak A$.
Observe that over frames every local formula $\alpha(x)$ can be rewritten as a two-way restricted formula, we just need to relativize quantifiers to the $r$-neighbourhood of $x$ for some $r$. Hence it suffices to show the following claim (in the finite):

\medskip

\noindent\textbf{Claim.}
Every first-order sentence invariant under disjoint unions is equivalent to a sentence of the form $\forall x \alpha(x)$, where $\alpha(x)$ is local.

\emph{Proof of the claim.} Suppose that $\theta$ is a first-order sentence invariant under disjoint unions in the finite. Then $\theta$ is in particular \emph{invariant under disjoint copies}, meaning that, for all finite
structures $\mathfrak{A}$ and for all $q\geq 1$,
$
\mathfrak A\models \theta
\Longleftrightarrow
q\cdot \mathfrak A\models \theta
$
where $q\cdot \mathfrak A$ denotes the disjoint union of $q$ copies of $\mathfrak A$.
Hence, by Proposition~19 of \cite{OTTO2004173}, $\theta$ is equivalent over finite structures to a Boolean combination of sentences of the form $\exists x\,\alpha(x)$, where $\alpha(x)$ is local.
Putting this Boolean combination into disjunctive normal form, we may write $\theta$ as
\[
 \bigvee_{i=1}^m (\Gamma_i^\forall \wedge \Gamma_i^\exists),
\]
where $\Gamma_i^\forall$ is a conjunction of formulas of the form $\forall x\alpha(x)$, and $\Gamma_i^\exists$ is a conjunction of formulas of the form $\exists x\,\beta(x)$, with $\alpha(x)$ and $\beta(x)$ local. Removing unsatisfiable disjuncts, we may assume that each
$
\Gamma_i^\forall \wedge \Gamma_i^\exists
$
is satisfiable in a finite structure.

We show that over finite structures, $\bigvee_{i=1}^m (\Gamma_i^\forall \wedge \Gamma_i^\exists)$ is equivalent to
$
\bigvee_{i=1}^m \Gamma_i^\forall.
$
One direction is immediate. For the converse, suppose that a finite structure $\mathfrak A$ satisfies
$
\Gamma_i^\forall
$
for some $i$. Since $\Gamma_i^\forall \wedge \Gamma_i^\exists$ is finitely satisfiable, choose a finite structure $\mathfrak B$ such that
$
\mathfrak B \models \Gamma_i^\forall \wedge \Gamma_i^\exists.
$
Then
$
\mathfrak A \uplus \mathfrak B \models \Gamma_i^\forall \wedge \Gamma_i^\exists,
$
because the universal conjuncts hold in both components, while the existential conjuncts are already witnessed in $\mathfrak B$. Hence
$
\mathfrak A \uplus \mathfrak B \models \bigvee_{i=1}^m (\Gamma_i^\forall \wedge \Gamma_i^\exists).
$
By invariance under disjoint unions, it follows that
$
\mathfrak A \models \bigvee_{i=1}^m (\Gamma_i^\forall \wedge \Gamma_i^\exists).
$
This shows that $\bigvee_{i=1}^m \Gamma_i^\forall$ implies $\bigvee_{i=1}^m (\Gamma_i^\forall \wedge \Gamma_i^\exists)$ over finite structures, and hence they are equivalent over finite structures.

Next we show that $\bigvee_{i=1}^m \Gamma_i^\forall$
is equivalent to one of its disjuncts. Suppose not. Then for each $i\in\{1,\dots,m\}$ there exists a finite structure $\mathfrak C_i$ such that
$\mathfrak C_i \models \bigvee_{j=1}^m \Gamma_j^\forall$
 but 
$\mathfrak C_i \not\models \Gamma_i^\forall.$
By invariance under disjoint unions, the disjoint union
$
 \biguplus_{i=1}^m \mathfrak C_i
$
also satisfies $\bigvee_{i=1}^m \Gamma_i^\forall$.
Thus, for some $k$,
$
 \biguplus_{i=1}^m \mathfrak C_i \models \Gamma_k^\forall.
$
Since $\Gamma_k^\forall$ is a conjunction of universal sentences of the form $\forall x \alpha(x)$, where $\alpha(x)$ is local, it is preserved under taking components of a disjoint union. Therefore
$
\mathfrak C_i \models \Gamma_k^\forall
$
for every $i$, in particular
$
\mathfrak C_k \models \Gamma_k^\forall,
$
contradicting the choice of $\mathfrak C_k$. It follows that
$
\bigvee_{i=1}^m \Gamma_i^\forall
$
is equivalent to $\Gamma_i^\forall$ for some $i$.
Finally, $\Gamma_i^\forall$ can be written in the form
$
\forall x \alpha(x),
$
where $\alpha(x)$ is local. This completes the proof of the claim.
\end{proof}

\vspace{-1.5ex}

\subsection{Bounded morphic images}

Let a \emph{p-formula} be a first-order formula obtained from atomic formulas (including equality statements) using conjunction, disjunction, existential and universal quantifiers, and bounded universal quantifiers of the form $\forall x(Ryx \to\cdots)$. A \emph{p-sentence} is a p-formula that is a sentence. An inductive argument shows that
p-sentences are preserved under taking images of bounded morphisms. In fact,
the converse holds as well, modulo logical equivalence: 

\begin{thmC}[\cite{vanBenthem1985}]
\label{thm:bmi}
A first-order sentence $\phi$ is preserved under surjective bounded morphisms iff $\phi$ is equivalent to a p-sentence.
\end{thmC}

We conjecture that this preservation theorem fails in the finite. 
Suggestive evidence for this is the fact
that Lyndon's positivity theorem (linking positivity to preservation under surjective homomorphisms) fails in the finite as well, as was orinally proved by Ajtai and Gurevich \cite{AjtaiGurevich87}.
More precisely, as pointed out in \cite{rossman2008homomorphism}, results of Ajtai and Gurevich \cite{AjtaiGurevich87} establish failure of Lyndon’s Positivity Theorem via a detour through circuit complexity. Namely, they showed that Monotone $\cap$ AC$^0$
$\neq$ Monotone-AC$^0$, that is, there is a (semantically) monotone Boolean function that is computable by AC$^0$ circuits,
but not by (syntactically) monotone AC$^0$
circuits. The failure of Lyndon’s Positivity theorem on finite structures follows via the descriptive complexity correspondence between AC$^0$ and first-order logic.
Recently, a much simpler and direct 
proof of the failure of Lyndon's positivity theorem in the finite was obtained by Kuperberg~\cite{kuperberg2023positive}. We suspect that Kuperberg's proof can be adjusted to prove that Theorem~\ref{thm:bmi} fails in the finite.
\vspace{-0.5ex}

\subsection{Ultrafilter extensions}

When it comes to preservation under taking ultrafilter extensions, the situation is quite different. Over arbitrary structures, there is no nice syntactic characterization of the FO formulas that are preserved under taking ultrafilter extensions. Indeed, preservation under ultrafilter extensions is   
    $\Sigma^1_1$-hard, which precludes the existence of a syntactic characterization in terms of a recursively enumerable class of formulas.
The crux of the $\Sigma^1_1$-hardness proof given in \cite{BalderPhD} lies in the following proposition.

\begin{table}[t]
\centering
\caption{Formula characterizing \((\mathbb{N}, <,\mathsf{Suc})\)}
\hrule
\begin{tabular}{ll}
\(\forall x\,\forall y\,\forall z\,(x < y \land y < z \rightarrow x < z)\)
    & (transitivity) \\

\(\forall x\,\forall y\,(x < y \lor y < x \lor x = y)\)
    & (trichotomy) \\

\(\forall x\,\exists y\,(x < y)\)
    & (unboundedness on the right) \\

\(\exists x\,\forall y\,\neg (y < x)\)
    & (boundedness on the left) \\

\(\exists x\,(x < x) \rightarrow \exists x y\,(x < x \wedge x < y \wedge \neg(y < y))\) & (reflexive before irreflexive) \\
\(\forall x\,\exists y\, \mathsf{Suc}(x,y)\) \\
\(\forall x\,\forall y\,(\mathsf{Suc}(x,y) \rightarrow x < y)\) \\
\(\forall x\,\forall y\,(\mathsf{Suc}(x,y) \rightarrow \forall z\,(x < z \rightarrow (y = z \lor y < z)))\)
\end{tabular}
\hrule
\label{tab:ue-encoding}
\end{table}

\begin{prop}[\cite{BalderPhD}]
Let \(\vartheta\) be the conjunction of the formulas displayed in Table~\ref{tab:ue-encoding}. For all models \(M\),
\[
M \cong (\mathbb{N}, <, \mathsf{Suc}) \quad\text{iff}\quad
M \models \vartheta\ \text{and}\ 
\mathrm{ue}M \not\models \vartheta.
\]
\end{prop}

In other words, $(\mathbb{N},<,\mathsf{Suc})$ is  uniquely characterized (up to isomorphism) 
by being a counterexample to preservation under ultrafilter extensions
for $\vartheta$. 

\begin{cor}
\label{cor:ue}
Let $N$ be a unary predicate, and 
let \(\varphi\) be any first-order formula preserved under ultrafilter
extensions (possibly containing the predicates $<$, $\mathsf{Suc}$, and $N$). The following are equivalent:
\begin{enumerate}
    \item \(\varphi\) has a model whose submodel defined by \(N\) is an expansion of
    \((\mathbb{N}, <, \mathsf{Suc})\);
    \item \(\varphi \wedge \vartheta\) is not preserved under ultrafilter extensions.
\end{enumerate}
where $\vartheta^N$ is the result of relativizing all quantifiers in $\vartheta$ by $N$.
\end{cor}

Corollary~\ref{cor:ue} can be used to prove that the set of first-order formulas
preserved under ultrafilter extensions is \(\Pi^1_1\)-hard by encoding the \(\Sigma^1_1\)-complete recurrent tiling problem~\cite{Harel85:recurring}.
We note that the unary relation symbol $N$ used in the above proof can easily be replaced by a binary relation. We leave it as an open question whether a single binary relation suffices for $\Sigma^1_1$-hardness.

Over finite structures, on the other hand, preservation under ultrafilter extensions trivializes, since the ultrafilter extension of a finite structure is always isomorphic to the structure in question.

\subsection{Where we stand}
Our analysis of transfer to the finite of preservation theorems for first-order logic has focused on semantic properties that are relevant to modal logic. The results can be described as a mixed bag. We did establish a new positive transfer result, namely for the Bisimulation Safety Theorem. On the other hand, 
most preservation theorem pertaining to operation on frames do not transfer to the finite. 
As in the modal case, we do not have an all-encompassing criterion for when  transfer to the finite is possible. 

In this final part of Section~\ref{sec:FO}, we just offer two general lines of thought. They apply to the frame setting  considered in Section~\ref{sec:FO}, but also to the model setting of Section~\ref{sec:modal}.

\vspace{1ex}

First, consider the \emph{complexity of defining} the semantic notions that we have studied, measured in the usual way in the Arithmetical Hierarchy. In the first-order case, consider the typical example of the  preservation theorem for submodels (but many examples would do). The condition that a formula $\phi$ is preserved under taking submodels can be stated equivalently as the validity of some effectively associated formula $\phi'$ describing the preservation. Since validity in FO is $\Sigma_0^{1}$, so is the preservation condition. The preservation theorem now says that this is equivalent to the existence of a formula $\psi$ with universal syntax such that $\phi \leftrightarrow \psi$ is valid. This second property is easily seen to be $\Sigma_0^{1}$ as well, and thus, the two sides in the preservation theorem match qua complexity of definition.

But now consider the restriction to finite models. Validity is obviously $\Pi_0^{1}$ there, but not $\Sigma_0^{1}$ (since it is undecidable by Trahtenbrot's Theorem). If we now match up the two sides of the same preservation theorem with the same formulas as before, the preservation condition restated as a validity is  $\Pi_0^{1}$. However, the claim that there exists a equivalent universal formula  introduces an existential quantifier that is in general irreducible, so we end up with the different complexity $\Sigma_0^{2}$. This mismatch alone may already explain why so many classical preservation theorems for first-order logic fail in the finite. 

But then, why do at least some preservation results transfer? At least in the prominent  cases that we mentioned, the point is this. The preservation theorem in the finite has an \emph{effective bound} on the potential definitions in the semantic format. For instance, in van Benthem's theorem the  equivalent modal formula can be taken to have a modal depth at most exponentially larger than the quantifier depth of the given first-order formula. But with such an effective bound, the  existential quantifier in the above $\Sigma_0^{2}$ form reduces to a finite disjunction of candidates, and we end up with $\Pi_0^{1}$ after all. The two sides  match. 

Could there be a  general result  that transfer of preservation theorems to the finite implies the existence of effective bounds? We leave this as an interesting  question to pursue.

\vspace{1ex}

Next, here is another general perspective  beyond the case-by-case analysis that we have given. There may also be a more positive challenge behind our negative results in the preceding sections. The noted failures applied to \emph{one particular formulation}  of the relevant preservation property that works for first-order logic. But as is shown in \cite{Barwise_vanBenthem_1999} for the case of Craig interpolation, transfer of meta-properties may depend on their formulation. Formulations that are equivalent for first-order logic may come apart for other  systems. 

To illustrate this way of thinking in the present setting, take the example  of Lyndon-style  preservation theorems for first-order logic (again, many examples would do).  The Lyndon theorem in first-order logic states that over arbitrary structures, first-order sentences with the relevant positive syntax  are characterized by being preserved under surjective homomorphisms. 
As is shown with a concrete language example in \cite{kuperberg2023positive}, the Lyndon theorem fails in the finite: the same semantic property no longer characterizes this formula class in the finite. 

But this negative observation leaves open the question whether, on finite models, this syntactic class is still semantically characterized by \emph{something stronger}, viz. a conjunction of (a) preservation under surjective homomorphisms, plus (b) some further natural semantic property guaranteed by the special positive syntax. In this way, analyzing know proofs of failure of classical preservation theorems in the finite might contain cues for extracting positive preservation theorems in the finite after all.\footnote{Some care is needed, since, to maintain the complexity match that we have discussed above, the added condition itself may have to be of $\Sigma_0^{2}$-complexity.}

Vice versa, looking at the Lyndon theorem in the opposite direction, one could also ask for a positive characterization of preservation under surjective homomorphisms in the finite in terms of generalizing the syntax used in the original preservation theorem.

\section{Computational aspects}
\label{sec:computational-aspects}

\medskip

What follows here is a brief exploration of some complexity-theoretic aspects of the preservation properties considered in this paper.  It can be treated as a digression.
\vspace{-2ex}

\subsection{The first-order case}

Consider the computational complexity of checking whether given formulas possess the semantic preservation properties  we have studied. We start with the first-order setting of Section~\ref{sec:FO}, turning to the modal setting of Section~\ref{sec:modal}  afterwards.

As in the modal case, many preservation properties can be encoded effectively as valid implications of a suitable sort, now inside the first-order language. This is immediate for preservation under homomorphisms using a simple syntactic relativization to a new unary predicate for preservation under submodels, and it can be done even for invariance under bisimulations, coding up the setting for a bisimulation in an obvious first-order manner.

By itself, the preceding observations tell us that testing for such properties is semi-decidable (RE), by the completeness theorem for first-order logic.
But what more can we say about the complexity?

We start with a result from \cite{ExploringLogicalDynamics}.
\begin{thm} It is undecidable if a given first-order formula is invariant for bisimulation.
\end{thm}

We present a slightly streamlined proof which  allows for further generalization. In what follows, for a pair of structures $M,N$, we will denote
by $[M,N]$ the ``model pair'' consisting of the disjoint union of $M$ and $N$ expanded with fresh unary predicates $P,Q$ denoting the two 
respective parts of the disjoint union. This is a standard construction in model theory.

\begin{proof} Take any first-order formula $\alpha$ which is not invariant for bisimulation (for instance, $\exists x Px$). Now choose a pointed model $(M, s)$ where $\alpha$ is true, and a bisimulation $Z$ from $(M, s)$ to a model $(N, t)$ with $\alpha$ false. Next, consider any first-order formula $\varphi$ whose vocabulary is disjoint from that of $\alpha$, choose two  new unary predicate letters $P, Q$ and define a new formula \#($\varphi$) with one free variable $u$ as the conjunction of the following three formulas:

\begin{enumerate}
    \item[(i)] $\exists x Px \land \exists x Qx \land \forall x(Px \lor Qx) \land \neg\exists x (Px \land Qx)$, \vspace{1mm}
\item[(ii)] $(\neg \varphi(x))^{P}$,\vspace{1mm}
\item[(iii)] $(\alpha(u))^{Q}$,
\end{enumerate}

\noindent where $\chi^P$ denotes the result of relativizing all quantifiers in $\chi$ by $P$, i.e., replacing $\exists x$ and $\forall x$ by
$\exists x(Px\land\cdots)$ and $\forall x(Px\to\cdots)$, respectively. The following equivalence holds:

\medskip

{\bf Claim} \quad $\varphi$ is valid if and only if \#($\varphi$) is invariant for bisimulation.

\medskip

From left to right, if $\alpha$ is valid, then \#($\varphi$) is equivalent to a contradiction $\bot$ which is trivially invariant for bisimulation. From right to left, suppose that $\varphi$ is not valid, and has some countermodel $K$. Now consider the two model pairs $[K, M], [K, N]$. We extend the given bisimulation $Z$ with the identity map on the domain of $K$ to obtain a bisimulation between $[K,M]$ and $[K,N]$. Now clearly, the model $[K, M]$ satisfies the above conjuncts (i), (ii). But $[K, N]$ does not since conjunct (iii) is false there. This shows that the formula \#($\varphi$) is not invariant for bisimulation.

Thus we have effectively reduced the validity problem for first-order logic to the problem of testing for bisimulation invariance. It follows that the latter problem is undecidable. \end{proof}

Inspecting this proof, only few assumptions are made, that can be stated as follows.

\begin{enumerate}
    \item[(a)] The preservation property to be tested does not hold for all formulas, and\vspace{1mm}
    \item[(b)] If the relevant semantic `transfer relation' holds between two models $M, N$, then 
it also holds between model pairs of the form $[K, M]$ and $[K, N]$.
\end{enumerate}

Let us now consider 
\emph{abstract preservation properties}
specified by a class $\mathcal{K}$ consisting 
of pairs of structures with the same signature (but where different pairs may have different signatures). We say
that a FO sentence $\phi$ \emph{satisfies} such an abstract
preservation property if, 
for all $(M,N)\in\mathcal{K}$ in the signature of $\phi$, it holds that 
$M\models\phi$ implies $N\models\phi$.
We say that $\mathcal{K}$ is
\emph{non-trivial} if not every
FO sentence satisfies it (corresponding to condition (a) above), and
that $\mathcal{K}$ is 
\emph{context-independent} if 
for all $(M,N)\in\mathcal{K}$
and for all structures $K$, also
$([K,M],[K,N])\in\mathcal{K}$ (corresponding to condition (b) above).

By the same proof as above, we then 
obtain the following more general undecidability result, inspired by Rice's theorem in recursion theory:

\begin{thm} For each abstract preservation properties that are non-trivial and context-independent,
testing if a FO sentence satisfies it is undecidable. \end{thm}

Examples of abstract preservation properties include: bisimulation-preservation, preservation under submodels, and preservation under homomorphic images. We leave it to the reader to verify that each of these is both non-trivial and
context-independent, as defined above.
It follows that
these are all undecidable. On the other
hand, for example, preservation 
under ultrafilter extensions does not naturally fit under the above regime (since taking the ultrafilter extension of a model
pair $[K,M]$ requires taking the ultrafilter extension of $K$) and
preservation under disjoint unions does not either (since this involves a relation between a set of structures and a structure, instead of a relationship between individual structures). We leave it for future research to find more general ways to phrase the result so that it also encompasses these cases.

In a sense, the preceding analysis has inverted the earlier style of thinking where preservation properties can be effectively encoded as first-order consequences. We have shown how, conversely,  arbitrary first-order consequences can be effectively encoded as preservation questions. Thus deciding these preservation properties inherits the complexity of first-order logic as a whole.

\vspace{-1ex}

\subsection{The modal case} 

Let us now turn to the more restricted modal realm. In Section~\ref{sec:modal}, we have seen many examples of semantic properties of modal formulas that can be effectively decided, due to the fact that they reduce  to testing a modal validity. 
In fact, it shows that such properties can be decided in PSpace, which is the computational complexity of validity, or satisfiability, for the basic (multi-)modal language. 
Below, we will provide a converse of sorts: we show that, for many semantic properties of modal formulas, testing them is as hard as testing for a modal validity.

\begin{defi} 
A \emph{semantic property} of  modal formulas is a property satisfying: 
\begin{enumerate} \item If two modal formulas $\phi, \psi$ are equivalent then $\phi$ has the property iff $\psi$ does. 
\item $\bot$ has the property. 
\item There are formulas without the property (``non-triviality'') 
\item Whenever two modal formulas $\phi,\psi$ have disjoint signatures and are both satisfiable, then $\phi\land\psi$ has the property if and only if both $\phi$ and $\psi$ have the property.
    \end{enumerate}
\end{defi}
Here, by the \emph{signature} of a modal formula we mean the set of propositional variables $p$ occurring in the formula as well as the set of 
modalities $a$ for which $\langle a\rangle$ or $[a]$ occurs in the formula.
It may be verified that all the semantic properties of modal formulas we have considered in Section~\ref{sec:modal} fit this template.

\begin{thm}
    Every semantic property of modal formulas is PSpace-hard to test.
\end{thm}

The proof is a one-line reduction from satisfiability: fix any formula $\chi$ that fails to have the 
property. Note that, by conditions 1 and 2, $\chi$ must be satisfiable. Then, a modal formula $\phi$, whose signature we may assume to be  disjoint from that of $\chi$, is
satisfiable iff $\chi\land \phi$ has the property in question.

\begin{rem}     The above proof uses multiple modalities. It
    could be modified to apply to  uni-modal formulas at the cost of complicating the definition of ``semantic property''. One way to do this is to (i) require that a formula $\phi$ has the property if and only if its syntactic relativisation $\phi^p$ does, where $p$ is a propositional variable not occurring $\phi$, and (ii) replace the use of the conjunction $\phi\land\psi$ by $\phi^p\land \psi^q$ (for fresh $p,q$), which ensures, intuitively, that the two conjuncts cannot ``interact'' with each other. \end{rem}

This concludes our brief exploration of complexity issues for preservation and related semantic properties. We leave more general results this to further investigation.

\section{Modal and first-order definability on  finite frames}
\label{sec:GbTh}

Finally, to broaden the scope of our comparative analysis for modal and first-order logic and provide some more perspective, we have added a further section on transfer to the finite of some other basic results from modal definability theory on frames. This will also highlight other basic features of working in the setting of finite frames, such as the many connections between logical languages and computational complexity classes.

\subsection{Goldblatt-Thomason theorems (characterizing modal definability)}

Recall that a
class $K$ of frames is 
\emph{modally definable} if there is a set  $\Phi$ of modal formulas such that
$K$ consists precisely of those frames
on which each modal formula in $\Phi$ is valid.
The following celebrated result of classical modal model theory tells us which first-order definable frame classes are modally definable.
Note that, as before, we restrict attention
to first-order sentences over a signature
consisting of binary relations only. 
Such sentences describe properties of 
(multi-modal) Kripke frames. 

\begin{thmC}[\cite{GoldblattThomason1975}]
Let ${\bf K}$ be any class of frames that is definable by a set of first-order sentences.
Then the following are equivalent:
\begin{enumerate}
    \item ${\bf K}$ is 
modally definable.
    \item ${\bf K}$ is closed under taking generated subframes disjoint unions and bounded morphic images, and  is complement is closed under taking ultrafilter extensions.
\end{enumerate}
\end{thmC}

The Goldblatt-Thomason theorem has a purely model-theoretic proof without detours into algebraic representation methods, cf. \cite{benthem1993:modal}. Using this proof method, the cited paper also notes that on finite transitive frames, closure under generated subframes, disjoint unions and bounded morphic images is necessary and sufficient for modal definability. 

In what follows, we analyze the situation slightly differently, using arbitrary frames but an extended modal language. We work with an \emph{extended modal language} containing both $\Diamond \varphi$ and the \emph{universal modality} $U\varphi$. 

First, consider any finite frame $F$, and take propositional variables $p_{\{x\}}$ for each point $x$ in $F$. Now expand $F$ to a model $M$ by interpreting these propositional variables as true only at their points. Then the following formula \$($F$) will be true in $M$:\medskip

Take one universal modality over a conjunction of \vspace{1ex}

(a) every point satisfies exactly one of the $p_{\{x\}}$, \vspace{1ex}

(b) for all $s, t$, if $s R t$ in $F$, then $p_{\{s\}} \rightarrow \Diamond p_{\{t\}}$, \vspace{1ex}

(c) for all $s, t$, if not $s R t$ in $F$, then $p_{\{s\}} \rightarrow \Box \neg p_{\{t\}}$
\medskip

Now the following statement holds for all models $(N, t)$:
\begin{lem}
$(N, t) \models \,\,$\$($F$) if and only if there is a bounded morphism (i.e., a functional bisimulation) from the model $N$ onto the model $M$.
\end{lem}

\begin{proof} 
The proof uses the obvious embedding sending all points satisfying the propositional variable $p_{\{x\}}$ in $N$ to the single point $x$ in $M$, where the clauses (b), (c) in the formula \$ enforce exactly the truth of the two back-and-forth clauses governing a bisimulation.\footnote{This method is similar to the analysis of finite models up to bisimulation in \cite{vanBenthem1998DynamicOdds}. }
\end{proof}

Using this, we can formulate the following version of the Goldblatt-Thomason Theorem for finite frames.

\begin{thm} A class ${\bf K}$ of finite frames is definable by a set of modal formulas in the language \{$\Box, U$\} if and only if ${\bf K}$ is closed under bounded morphic images.

\end{thm}

\begin{proof} We show that the class ${\bf K}$ is defined by its own modal theory in our language. Consider any frame $F$ which satisfies this theory, expand it to a model $M$ as described above, and construct the formula \$($F$). Obviously, the formula $\neg$\$($F$) is not valid in the set ${\bf K}$ (since $F$ satisfies the modal theory of ${\bf K}$) , so there exists at least one frame $G$ in ${\bf K}$ which refutes it for some valuation. But then, using the preceding observation, it follows that $F$ is a bounded morphic image of $G$, and therefore $F$ belongs to the class ${\bf K}$. 
\end{proof}

Of course, this is a brute-force method, producing in principle, a set of modal formulas defining the given class. Thus, obvious questions remain about the scope of our analysis. 

One is when there is \emph{one single modal formula} defining the given set ${\bf K}$. We suspect that the assumption that ${\bf K}$ be definable by one first-order sentence suffices for this, but we have no proof.

Another immediate  question is this. Can the above result be extended to work for just the basic modal language without the universal modality?

Our final question is as follows. The preceding  results were about {\it extending} the basic modal language. What about the converse direction: What happens to the preceding analysis when we {\it restrict} the basic modal language to, say,  positive fragments of various sorts? 

\vspace{-3ex}

\subsection{Modal correspondence theory (characterizing first-order definability)}\label{subsec:McKinsey}

In Section~\ref{sec:modal}, we have seen that many semantic properties of modal formulas
(monotonicity, preservation under induced substructures, complete additivity, strong continuity)
hold in the finite if and only if they hold over arbitrary  Kripke
structures. The same is not always true. In particular, it  does \emph{not} hold for the property of \emph{defining a first-order frame condition}.

Let us first recapitulate some known facts. \emph{Frame validity} of a modal formula $\varphi$ is in principle a monadic universal second-order property, obtained by prefixing the first-order standard translation ST($\varphi$) with universal quantifiers over its unary predicate letters. However, many modal formulas have corresponding properties definable in first-order logic, and a large class of these is described by the Sahlqvist Theorem (cf.~\cite{BdRV}). 
A famous example of a natural modal formula
that corresponds to a non-first-order property of frames is L\"ob's Axiom $\Box(\Box p \rightarrow p) \rightarrow \Box p)$  in provability logic, which defines the class of transitive and upward well-founded relational frames. 

Let us now turn to correspondence on finite frames. Here is a preliminary observation.

\begin{rem}[A matter of definability] It is a well known fact that, for modally definable frames classes, 
\emph{being definable by a set of FO sentences} implies \emph{being definable by a single FO sentence}. On finite frames, on the other hand, \emph{being definable by a set of FO sentences} trivializes: every isomorphism-closed class of finite frames is definable by a set of first-order sentences.
Indeed, it suffices to take, for each finite frame not in the class, the negation of the first-order sentence describing it up to isomorphism. The resulting set of sentences defines the given class. However, in the finite realm, it does remain meaningful to study definability by a single FO sentence.
\end{rem}

Restricted to finite frames, correspondence results of the preceding type may simplify, since some second-order modal axioms become first-order definable. For instance, L\"ob's Axiom now defines the transitive irreflexive frames (a first-order definable class).
\begin{thm}
For finite Kripke frames $\mathfrak{F}=(W,R)$, it holds that  $\mathfrak{F}\models \Box(\Box p \to p)\to \Box p$  iff $R$ is transitive and irreflexive.
\end{thm}

On the other hand, the McKinsey Axiom $\Box\Diamond p\to \Diamond\Box p$, whose non-first-orderness was shown in \cite{Benthem75:note}, still lacks a first-order correspondent even when we restrict attention to finite frames. We prove this by showing that, for every $n$, there exist two finite frames that satisfy the same first-order sentences of quantifier rank at most $n$, but differ on the validity of the McKinsey Axiom.

\begin{defi}
Let $n \ge 0$ be an integer. We define a finite frame
$
\mathfrak{F}_n = (W,R)
$
as follows.
The domain is
\[
W = \{r\} \cup \{a_i : 0 \le i \le n\} \cup \{b_i : 0 \le i \le n\},
\]
where all elements are assumed to be distinct.
The relation $R \subseteq W \times W$ is defined by:
\begin{enumerate}
  \item $(r,a_i) \in R$ for all $0 \le i \le n$;
  \item $(a_i,b_j) \in R$ iff $j=i$ or $j=i+1$, where the index $i+1$ is understood to be $0$ when $i=n$.
  \item $(b_i,b_i) \in R$ for all $0 \le i \le n$.
\end{enumerate}
No other pairs belong to $R$.

We further define a relational structure $\mathfrak{G}_n = (W,R',P)$ on the same domain $W$, where $R'$ is a binary relation and $P$ is a unary relation.
The relation $R'$ is obtained from $R$ in $\mathfrak{F}_n$ by removing
all edges of the form $(r,a_i)$, that is,
\[
R' = R \setminus \{(r,a_i) : 0 \le i \le n\}.
\]
The unary relation $P$ is interpreted as the singleton $\{r\}$.
\end{defi}

For example, The frames $\mathfrak{F}_5$ and $\mathfrak{G}_5$ are depicted below. Solid arrows indicate the relation $R$, while dashed arrows indicate the relation $R'$. The dashed circle marks the unique element satisfying the unary predicate $P$.
\tikzset{
  Rstyle/.style={->, >=Stealth, thick},
  Rprimestyle/.style={->, >=Stealth, thick, dashed}
}

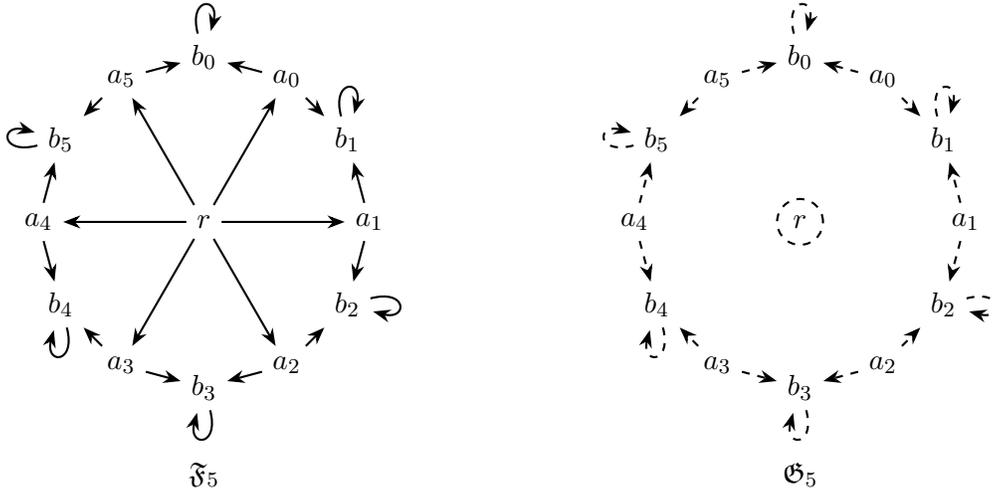
\begin{figure}[htbp]
\centering

\begin{minipage}{0.48\linewidth}
\centering
\begin{tikzpicture}[Rstyle]
\node (r) at (0,0) {$r$};

\def\radius{2.2}

\node (b0) at (90:\radius) {$b_0$};
\node (b1) at (30:\radius) {$b_1$};
\node (b2) at (-30:\radius) {$b_2$};
\node (b3) at (-90:\radius) {$b_3$};
\node (b4) at (-150:\radius) {$b_4$};
\node (b5) at (150:\radius) {$b_5$};

\node (a0) at (60:\radius) {$a_0$};
\node (a1) at (0:\radius) {$a_1$};
\node (a2) at (-60:\radius) {$a_2$};
\node (a3) at (-120:\radius) {$a_3$};
\node (a4) at (180:\radius) {$a_4$};
\node (a5) at (120:\radius) {$a_5$};

% r -> a_i  (only in F_5)
\foreach \i in {0,...,5} {
  \draw (r) -- (a\i);
}

% a_i -> adjacent b's
\draw (a0) -- (b0);
\draw (a0) -- (b1);
\draw (a1) -- (b1);
\draw (a1) -- (b2);
\draw (a2) -- (b2);
\draw (a2) -- (b3);
\draw (a3) -- (b3);
\draw (a3) -- (b4);
\draw (a4) -- (b4);
\draw (a4) -- (b5);
\draw (a5) -- (b5);
\draw (a5) -- (b0);

% reflexive loops on b_i
\draw (b0) edge[loop above] ();
\draw (b1) edge[loop above] ();
\draw (b2) edge[loop right] ();
\draw (b3) edge[loop below] ();
\draw (b4) edge[loop below] ();
\draw (b5) edge[loop left] ();

\end{tikzpicture}

{ $\mathfrak{F}_5$}
\end{minipage}
\hfill
\begin{minipage}{0.48\linewidth}
\centering
\begin{tikzpicture}[Rprimestyle]
\node[draw, circle] (r) at (0,0) {$r$};

\def\radius{2.2}

\node (b0) at (90:\radius) {$b_0$};
\node (b1) at (30:\radius) {$b_1$};
\node (b2) at (-30:\radius) {$b_2$};
\node (b3) at (-90:\radius) {$b_3$};
\node (b4) at (-150:\radius) {$b_4$};
\node (b5) at (150:\radius) {$b_5$};

\node (a0) at (60:\radius) {$a_0$};
\node (a1) at (0:\radius) {$a_1$};
\node (a2) at (-60:\radius) {$a_2$};
\node (a3) at (-120:\radius) {$a_3$};
\node (a4) at (180:\radius) {$a_4$};
\node (a5) at (120:\radius) {$a_5$};

% a_i -> adjacent b's
\draw (a0) -- (b0);
\draw (a0) -- (b1);
\draw (a1) -- (b1);
\draw (a1) -- (b2);
\draw (a2) -- (b2);
\draw (a2) -- (b3);
\draw (a3) -- (b3);
\draw (a3) -- (b4);
\draw (a4) -- (b4);
\draw (a4) -- (b5);
\draw (a5) -- (b5);
\draw (a5) -- (b0);

% reflexive loops on b_i
\draw (b0) edge[loop above] ();
\draw (b1) edge[loop above] ();
\draw (b2) edge[loop right] ();
\draw (b3) edge[loop below] ();
\draw (b4) edge[loop below] ();
\draw (b5) edge[loop left] ();

\end{tikzpicture}

{ $\mathfrak{G}_5$}
\end{minipage}

\caption{The frames $\mathfrak{F}_5$ and $\mathfrak{G}_5$.}

\label{fig:parity-frames}

\end{figure}

\begin{lem}\label{lem:mckinsey-parity}
The McKinsey axiom $\Box\Diamond p \to \Diamond\Box p$ is valid in $\mathfrak{F}_n$
if and only if $n$ is even.
\end{lem}

\begin{proof}
\noindent($\Rightarrow$) Assume that $n$ is odd. We show that the axiom is not valid in $\mathfrak{F}_n$. Define a valuation $V$ by
$
V(p)=\{ b_i : i \text{ is even and } 0\leq i\leq n\}.
$
We claim that $\mathfrak{F}_n,V,r \models \Box\Diamond p$ but $\mathfrak{F}_n,V,r \not\models \Diamond\Box p$.

Indeed, the successors of $r$ are exactly the points $a_i$ for $i=0,\dots,n$.
Each $a_i$ has exactly two successors, namely $b_i$ and $b_{i+1}$,
where $b_{i+1}$ is understood as $b_0$ when $i = n$.
Since $n$ is odd, exactly one of $b_i$ and $b_{i+1}$ has odd index.
Hence $\mathfrak{F}_n,V,a_i \models \Diamond p\wedge \Diamond \neg p$ for all $i (0\leq i\leq n)$, so $\mathfrak{F}_n,V,r \models \Box\Diamond p$ and $\mathfrak{F}_n,V,r \not\models \Diamond\Box p$.
This refutes the McKinsey axiom at $r$.

\noindent($\Leftarrow$) 
Assume that $n$ is even, and let $V$ be an arbitrary valuation. 
It is straightforward to verify that $\mathfrak{F}_n,V,w \models \Box\Diamond p \to \Diamond\Box p$
for every $w \in W\setminus\{r\}$. Hence it suffices to show that
\[
\mathfrak{F}_n,V,r \models \Box\Diamond p \to \Diamond\Box p .
\]
Suppose, towards a contradiction, that $\mathfrak{F},V,r \models \Box\Diamond p\wedge \Box\Diamond\neg p$. Then we have $a_i \models \Diamond p \wedge \Diamond \neg p$  for every $0\leq i \leq n$. By the definition of $R$, the only successors of $a_i$ are $b_i$ and $b_{i+1}$. Hence, for each $0 \le i \le n$, among the pair $\{b_i,b_{i+1}\}$, at least one satisfies $p$ and at least one satisfies $\neg p$; equivalently,
\begin{equation}\label{eq:split}
\mathfrak{F}_n, V, b_i \models p \quad\Longleftrightarrow\quad \mathfrak{F}_n, V, b_{i+1} \models \neg p .
\end{equation}

Now if $\mathfrak{F}_n, V, b_0 \models p$, then by repeatedly applying~\eqref{eq:split} we obtain
$
\mathfrak{F}_n, V, b_i \models p \ \text{iff}\ i \ \text{is even}.
$
Since $n$ is even, it follows that $\mathfrak{F}_n, V, b_n \models p$.
Moreover, by the definition of $R$, $a_n$ has no successor other than $b_n$ and $b_0$, and therefore
$
\mathfrak{F}_n, V, a_n \models \Box p .
$
This contradicts $\mathfrak{F}_n, V, a_n \models \Diamond \neg p$.
In the case of $\mathfrak{F}_n,V,b_0 \models \neg p$,
the same propagation argument yields $\mathfrak{F}_n,V,b_n \models \neg p$, and hence
$\mathfrak{F}_n,V,a_n \models \Box \neg p$, contradicting $\mathfrak{F}_n,V,a_n \models \Diamond p$.
Therefore,
$
\mathfrak{F}_n, V, r \models \Box\Diamond p \rightarrow \Diamond\Box p.
$ Since $V$ was arbitrary, we conclude that the McKinsey axiom is valid in $\mathfrak{F}_n$.
\end{proof}
    
In the following we show that the McKinsey axiom is not first-order definable over finite frames by a locality argument based on Hanf's locality theorem; see~\cite{EF} for further details on locality theorems.

To state Hanf's theorem, we recall the notions of Gaifman graph and $r$-neighbourhood (see Section~\ref{subsec:disjoint unions}), and use the following notion of $r$-neighbourhood type.
Two elements $u$ and $v$ are said to have the same $r$-neighbourhood type if their $r$-neighbourhoods are isomorphic via an isomorphism mapping $u$ to $v$.

Informally, Hanf's theorem asserts that two structures cannot be distinguished by first-order sentences if, for every neighbourhood type, they contain either the same number of elements of that type or sufficiently many elements of that type.

\begin{thmC}[Hanf's Theorem]
    For every first-order sentence $\phi$, there exists $r \in \mathbb{N}$ such that for every $e\in \mathbb{N}$, if structures $\mathfrak A$ and $\mathfrak B$ satisfy the following conditions, then $\mathfrak A\models \phi$ iff $\mathfrak B\models \phi$:
\begin{enumerate}
    \item every $r$-neighbourhood in $\mathfrak{A}$ and $\mathfrak{B}$ has size less than $e$
    \item for each $r$-neighbourhood type $\iota$, either
    \begin{itemize}
        \item $\mathfrak A$ and $\mathfrak B$ contain the same number of elements whose $r$-neighbourhood is of type $\iota$, or
        \item both $\mathfrak A$ and $\mathfrak B$ contain more than $re$ such elements.
    \end{itemize}
\end{enumerate}
\end{thmC}

We now apply Hanf's Theorem to obtain the following lemma.

\begin{lem}\label{lem:local}
For every first-order sentence $\phi$, there exists an even number $n$ such that $\mathfrak{G}_{n} \models \phi$ iff $\mathfrak{G}_{n+1} \models \phi$.
\end{lem}

\begin{proof}
Fix a first-order sentence $\phi$. By Hanf's Theorem, there exists $r\in \mathbb{N}$ such that the truth of $\phi$ depends only on the numbers of realizations of each $r$-neighbourhood type, counted up to a threshold of the form $re$.

Choose $e > 2r + 1$, and let $n$ be an even integer such that $n > 2re$.
We verify that $\mathfrak{G}_n$ and $\mathfrak{G}_{n+1}$ satisfy the conditions of Hanf's Theorem.

In $\mathfrak G_n$ (and likewise in $\mathfrak G_{n+1}$), there are exactly three
$r$-neighbourhood types:
\begin{itemize}
  \item the $r$-neighbourhood of the isolated point $r$;
  \item the $r$-neighbourhood of a point $a_i$ $(0\le i\le n)$;
  \item the $r$-neighbourhood of a point $b_i$ $(0\le i\le n)$).
\end{itemize}
We denote these types by $\iota_1$, $\iota_2$, and $\iota_3$, respectively.

By the construction, every $r$-neighbourhood in $\mathfrak G_n$ and $\mathfrak G_{n+1}$ has size at most $2r+1$, which is less than $e$. So the first condition of Hanf's Theorem is satisfied.
For the second condition, since
$
n> 2re,
$
it follows that, for each neighbourhood type $\iota \in \{\iota_2,\iota_3\}$, both $\mathfrak{G}_n$ and $\mathfrak{G}_{n+1}$ contain more than $re$ elements of type $\iota$.
Therefore, all conditions of Hanf's Theorem are satisfied, and we conclude that $\mathfrak{G}_n$ and $\mathfrak{G}_{n+1}$ satisfy the same first-order sentences of quantifier rank at most $r$.
\end{proof}

\begin{thm}
There is no FO sentence $\phi$ such that 
for all finite frames $\mathfrak{F}$, $\mathfrak{F}\models \Box\Diamond p\to \Diamond\Box p$ iff $\mathfrak{F}\models\phi$.
 \end{thm}

\begin{proof}
We first define a first-order interpretation $\mathcal{I}$ that transforms $\mathfrak{G}_k$ into $\mathfrak{F}_k$ for $k \in \mathbb{N}$.
The domain formula is $\phi_{\mathrm{dom}}(x) \equiv \top$, and the interpreted binary relation $R$ is defined by
\[
\phi_R(x,y) := R'(x,y) \;\vee\; \bigl(P(x) \wedge \neg P(y) \wedge \neg R'(y,y)\bigr).
\]
Intuitively, $\phi_R$ keeps all $R'$-edges and, in addition, connects the unique element satisfying $P$ to every non-reflexive element.
In $\mathfrak{G}_k$, the predicate $P$ holds only at the isolated point $r$, and the non-reflexive elements are exactly the points $a_i$.
Hence, $\phi_R$ adds precisely the edges $(r,a_i)$, and therefore $\mathcal{I}(\mathfrak{G}_k)=\mathfrak{F}_k$.

Next, we relate the first-order theories of $\mathfrak{F}_k$ and $\mathfrak{G}_k$.
Let $\psi$ be a first-order sentence in the language of $\mathfrak{F}_k$, that is, over a single binary relation symbol $R$.
Let $\psi^*$ be the first-order sentence in the language of $\mathfrak{G}_k$ obtained by replacing each atomic subformula $R(u,v)$ in $\psi$ with $\phi_R(u,v)$.
Then for every $k\in\mathbb{N}$
\[
\mathfrak{F}_k \models \psi \iff\mathfrak{G}_k \models \psi^*.
\]

Now suppose, towards a contradiction, that there exists a first-order sentence $\theta$ that defines the McKinsey axiom on finite frames.
Let $\theta^*$ be the sentence in the language of $\mathfrak{G}_k$ obtained from $\theta$ by replacing each atomic subformula $R(u,v)$ with $\phi_R(u,v)$.
By Lemma~\ref{lem:local}, there exists an even number $n$ such that
$
\mathfrak{G}_{n} \models \theta^* \iff \mathfrak{G}_{n+1} \models \theta^*.
$
Using the interpretation equivalence above, it follows that
$
\mathfrak{F}_{n} \models \theta \iff \mathfrak{F}_{n+1} \models \theta.
$
On the other hand, by Lemma~\ref{lem:mckinsey-parity}, we have
$
\mathfrak{F}_n \vDash \Box\Diamond p \to \Diamond\Box p
$
and 
$
\mathfrak{F}_{n+1} \nvDash \Box\Diamond p \to \Diamond\Box p,
$
since $n$ is even and $n+1$ is odd.
As $\theta$ defines the McKinsey axiom on finite frames, this implies
$
\mathfrak{F}_n \models \theta
$
and
$
\mathfrak{F}_{n+1} \not\models \theta,
$
a contradiction.
\end{proof}

\vspace{-3ex}

\subsection{Modal corrrespondence hierarchy over finite frames}
From the preceding section, we know that McKinsey's Axiom $\Box\Diamond p \to \Diamond\Box p$ is not first-order definable over finite frames, while L\"ob's Axiom $\Box(\Box p \to p) \to \Box p$ is.
This suggests that, even when we restrict attention to finite frames, different modal formulas may occupy different
levels of definability strength.
This leads us to consider
the hierarchy of relational properties
that arises inside modal logic.

We briefly step back and recall the situation over arbitrary frames. Recall that the frame validity of a modal 
formula is expressible in monadic second-order logic (MSO). Inside MSO lies a hierarchy of
increasingly expressive logics:
\[ \text{FO} \subsetneq \text{FO+monTC} \subsetneq \text{FO+monLFP} \subsetneq \text{MSO}\]
where FO+monTC is the extension of first-order logic with an operator $[TC_{x,y}~\phi](u,v)$ for taking the transitive closure of definable binary relations; and FO+monLFP (also known as $\mu$FO) is the extension of 
first-order logic with an operator $[LFP_{X,x}~\phi](u)$ for taking the least fixed point of a definable operation on sets. 
Different modal axioms naturally fall in different levels of this hierarchy, intuitively corresponding to different degrees of non-first-orderness:
\begin{itemize}
\item The conjunction of the two modal axiom $[b]p \leftrightarrow p\land [a][b]p$ and $p\land [b](p\to [a]p)\to [b]p$
is valid on a frame if and only if $R_b$ is the reflexive-transitive closure of $R_a$. This property is not first-order  but is definable in FO+monTC. 
\item The L\"ob axiom falls inside FO+monLFP, as follows from its syntactic shape~\cite{Benthem2005:minimal}. 
The frame property it defines (i.e., transitivity and converse well-foundedness) is known to be  undefinable in the infinitary language  $L_{\omega_1\omega}$~\cite{Lopez1966}, 
and hence is also not expressible 
in
FO+monTC, since $\text{FO+monTC} \subseteq L_{\omega_1\omega}$.
\item The McKinsey axiom is `genuinely' MSO: the frame property it defines is not expressible in FO+monLFP, as was 
shown in~\cite{Benthem2005:minimal}.\footnote{This highest level may have a lot of further structure, given the encoding of consequence for monadic $\Pi^{1}_{1}$-sentences on standard models as frame consequence between their translations into modal formulas~\cite{Thomason75}.}
\end{itemize}

We now turn to the finite setting.
To begin with, the modal formula discussed above, expressing reflexive-transitive closure,
already provides an example of a frame condition that is definable in
\(\mathrm{FO+monTC}\) but not in \(\mathrm{FO}\),
even over finite frames.

To analyse the remaining levels, we appeal to descriptive complexity 
(see \cite{Immerman1999,graedel2007finite}).
Over finite structures, it is known that
\[
\mathrm{FO+monTC} \subseteq \mathbf{NL},
\qquad
\mathrm{FO+monLFP} \subseteq \mathbf{P}.
\]
More precisely, every class of finite structures definable
in $\mathrm{FO+monTC}$ is decidable in nondeterministic logarithmic space,
and every class definable in $\mathrm{FO+monLFP}$ is decidable in polynomial time. 

We use this observation to obtain non-definability results. Here is a striking instance:  

\vspace{-2ex}

\subsection{Deciding the McKinsey axiom is $\mathbf{coNP}$-complete.}
\begin{thm}\label{thm:mckinsey-conp}
The problem of deciding whether a finite frame validates
$
\Box\Diamond p \to \Diamond\Box p
$
is $\mathbf{coNP}$-complete.
\end{thm}

\begin{proof}
To see that the problem lies in $\mathbf{coNP}$, it suffices to show that non-validity is in $\mathbf{NP}$.
Observe that
$
\mathfrak F \not\models \Box\Diamond p \to \Diamond\Box p
$
iff there exist a world $w \in W$ and a valuation $V$
such that
$
\mathfrak F,V,w \models \Box\Diamond p\wedge \Box\Diamond \neg p.
$
A nondeterministic algorithm can guess $(w,V)$
and verify in polynomial time whether
$\Box\Diamond p \wedge \Box\Diamond \neg p$
holds at $w$ under $V$.
Hence non-validity is in $\mathbf{NP}$,
and validity is in $\mathbf{coNP}$.

For $\mathbf{coNP}$-hardness, we reduce from the \textsc{Set Splitting} problem to non-validity of $\Box\Diamond p \to \Diamond\Box p$.
Recall that \textsc{Set Splitting} asks,
given a finite set $S$ and a family
$\mathcal F \subseteq \mathcal P(S)$,
whether there exists a partition
$S_1 \cup S_2=S$ with $S_1 \cap S_2 = \varnothing$
such that every $f \in \mathcal F$
intersects both $S_1$ and $S_2$.
This problem is $\mathbf{NP}$-complete.

Given an instance $(S,\mathcal F)$,
we construct in polynomial time
a finite frame
$\mathfrak F_{S,\mathcal F}=(W,R)$ as follows.
Let
\[
W \ :=\ \{w\} \ \cup\ \mathcal F \ \cup\ S,
\]

\vspace{1mm}
\noindent where all elements are taken to be distinct worlds.
The accessibility relation $R$ is defined by:
\begin{itemize}
\item $wRf$ for all $f \in \mathcal F$;
\item for $f \in \mathcal F$ and $s \in S$, $fRs$ iff $s \in f$ (recall that $f \subseteq S$);
\item $sRs$ for all $s \in S$.
\end{itemize}
No further edges are present.

Note that every world $u \neq w$ has exactly one $R$-successor.
Hence $\Box\Diamond p \to \Diamond\Box p$ is valid at every such world.
Therefore failure of the formula, if any, must occur at $w$.

We claim that $(S,\mathcal F)$ is a yes-instance of \textsc{Set Splitting}
iff
$
\mathfrak F_{S,\mathcal F}
\not\models
\Box\Diamond p \to \Diamond\Box p.
$

\smallskip
\noindent
($\Rightarrow$)
Suppose $S=S_1 \cup S_2$ is a valid splitting.
Define a valuation $V$ by setting $V(p) = S_1$.

For every $f \in \mathcal F$, the splitting property yields $f \cap S_1 \neq \varnothing$.
Hence there exists $s \in f \cap S_1$
with $fRs$ and $s \models p$,
so $f \models \Diamond p$.
Since $w$ sees exactly all $f \in \mathcal F$,
we obtain
$
\mathfrak F_{S,\mathcal F},V,w \models \Box\Diamond p.
$
On the other hand,
because $f \cap S_2 \neq \varnothing$,
there exists $s \in f \cap S_2$ with $fRs$ and $s \not\models p$.
Thus $f \not\models \Box p$.
Consequently no $f$ satisfies $\Box p$, and therefore
$
\mathfrak F_{S,\mathcal F},V,w
\not\models \Diamond\Box p.
$

\smallskip
\noindent
($\Leftarrow$)
Conversely,
assume
$
\mathfrak F_{S,\mathcal F}
\not\models
\Box\Diamond p \to \Diamond\Box p.
$
Then there exists a valuation $V$ such that
$
\mathfrak F_{S,\mathcal F},V,w \models \Box\Diamond p$
and 
$
\mathfrak F_{S,\mathcal F},V,w \not\models \Diamond\Box p.
$

Define
$
S_1 := V(p)\cap S$
and 
$
S_2 := S \setminus S_1.
$
From $w \models \Box\Diamond p$ it follows that for every $f \in \mathcal F$ there exists $s \in f$ with $s \in S_1$.
From $w \not\models \Diamond\Box p$ it follows that for every $f \in \mathcal F$ there exists $s \in f$ with $s \in S_2$.
Thus each $f$ meets both $S_1$ and $S_2$, so $(S_1,S_2)$ is a valid splitting.

\medskip
We have  established a polynomial-time reduction from \textsc{Set Splitting}
to non-validity of $\Box\Diamond p \to \Diamond\Box p$.
Since \textsc{Set Splitting} is $\mathbf{NP}$-complete,
non-validity is $\mathbf{NP}$-hard, and validity is $\mathbf{coNP}$-hard.
Combining hardness with membership in $\mathbf{coNP}$
yields $\mathbf{coNP}$-completeness.
\end{proof}

From Theorem~\ref{thm:mckinsey-conp} and the fact that
$\mathrm{FO+monLFP}$-definable classes of finite structures lie in $\mathbf{P}$, we obtain the following corollary.
\begin{cor}
The class of finite frames defined by the McKinsey axiom $\Box\Diamond p \to \Diamond\Box p$ is not definable
in $\mathrm{FO+monLFP}$ unless $\mathbf{P}=\mathbf{NP}$.
\end{cor}
\vspace{-3ex}

\subsection{A $\mathbf{P}$-complete modal formula.}
\
We now turn to a modal formula that is definable in
$\mathrm{FO+monLFP}$ but not in $\mathrm{FO+monTC}$,
unless $\mathbf{NL}=\mathbf{P}$. Define
\[
\phi_h \ :=\ 
\Diamond(\Box^- t \wedge \Box^+ \neg t) \wedge
(\Diamond^+ t \rightarrow \Box^+ t).
\]
We will see that checking whether $\phi_h$ holds at a world $s$ in a given finite model amounts to solving a Horn-UNSAT instance.

\begin{thm}\label{thm:horn-p}
Deciding whether a finite frame validates $\phi_h$ is $\mathbf{P}$-hard under logspace many-one reductions.
\end{thm}

\begin{proof}
We reduce from the $\mathbf{P}$-complete problem \textsc{Horn-UNSAT}.
Let $\alpha = C_1 \wedge \dots \wedge C_n$ be a Horn formula over variables 
$\mathsf{Var}(\alpha)=\{P_1,\dots,P_m\}$, where each clause $C_j$ is of the form
\[
\neg P_{i_1} \vee \dots \vee \neg P_{i_\ell} \vee P_h,
\]
with at most one positive literal.
The problem \textsc{Horn-UNSAT} asks whether $\alpha$ is unsatisfiable. It is well known that \textsc{Horn-UNSAT} is $\mathbf{P}$-complete under logspace reductions.

We define a logspace-computable reduction
$
\alpha \mapsto \mathfrak F_\alpha ,
$
where $\mathfrak F_\alpha = (W_\alpha, R, R^+, R^-)$ is a finite frame defined as follows.
The domain consists of clause and variable nodes:
\[
W_\alpha = \{c_1,\dots,c_n\} \cup \{p_1,\dots,p_m\}.
\]
The relation $R$ connects every world to every clause node:
\[
w R c_j \quad \text{for all } w \in W_\alpha \text{ and } 1 \le j \le n,
\]
and there are no other $R$-edges.
For each clause $C_j$, let $\mathrm{Pos}(C_j)$ and $\mathrm{Neg}(C_j)$ denote the sets of variables occurring positively and negatively in $C_j$, respectively.
We define the relations $R^+$ and $R^-$ by
\[
c_j R^+ p_i \text{ iff } P_i \in \mathrm{Pos}(C_j), \qquad
c_j R^- p_i \text{ iff } P_i \in \mathrm{Neg}(C_j),
\]
and there are no other pairs in $R^+$ or $R^-$.

The mapping $\alpha \mapsto \mathfrak F_\alpha$ is computable in logarithmic space.
The size of $\mathfrak F_\alpha$ is polynomial in $|\alpha|$.
The relation $R$ is trivial to compute, since $(w,c_j) \in R$ holds for every $w \in W_\alpha$ and every clause node $c_j$, and there are no other $R$-edges.
To determine whether a pair $(c_j, p_i)$ belongs to $R^+$ or $R^-$, it suffices to inspect the clause $C_j$
and check whether $P_i$ or $\neg P_i$ occurs in it.
This can be done using only logarithmic space to store the indices $j$ and $i$,
while the input formula $\alpha$ can be re-scanned as needed.

\medskip
\noindent\textbf{Claim.}
For every Horn formula $\alpha$, 
$\alpha$ is unsatisfiable if and only if 
$\mathfrak F_\alpha \models \phi_h$.

\smallskip
\noindent\emph{Proof of Claim.}
\smallskip
\emph{($\Rightarrow$)}
Assume $\alpha$ is unsatisfiable. 
Let $V$ be any valuation for $t$ on $\mathfrak F_\alpha$ and $w\in W_\alpha$.
Define the propositional assignment $I_V$ by
\[
I_V(P_i)=\mathsf{true}
\quad\text{iff}\quad
\mathfrak F_\alpha,V,p_i\models t.
\]
Since $\alpha$ is unsatisfiable, $I_V\not\models\alpha$.
Thus some clause $C_k := \neg P_{i_1}\vee\dots\vee\neg P_{i_\ell}\vee P_h$ is falsified by $I_V$.
This means that
$I_V(P_{i_r})=\mathsf{true}$ for all $r$ and
$I_V(P_h)=\mathsf{false}$.
By construction of $R^-$ and $R^+$, all $R^-$-successors of $c_k$ satisfy $t$, and every $R^+$-successor satisfies $\neg t$.
Hence
\[
\mathfrak F_\alpha,V,c_k\models
\Box^- t \wedge \Box^+\neg t.
\]
Since $wRc_k$, we obtain
$
\mathfrak F_\alpha,V,w\models
\Diamond(\Box^- t \wedge \Box^+\neg t).
$
We then have
$\mathfrak F_\alpha\models\Diamond(\Box^- t \wedge \Box^+\neg t)$, as $V$ and $w$ were arbitrary.
Moreover, each clause contains at most one positive literal,
so every world has $R^+$-outdegree at most one.
Hence $\Diamond^+ t\to\Box^+ t$ is valid on $\mathfrak F_\alpha$.
Therefore $\mathfrak F_\alpha\models\phi_h$.

\smallskip
\emph{($\Leftarrow$)}
We prove the contrapositive.
Assume that $\alpha$ is satisfiable, i.e.\ there exists a propositional assignment $I:\mathsf{Var}(\alpha)\to\{\mathsf{true},\mathsf{false}\}$ such that $I\models\alpha$.
We show that $\mathfrak F_\alpha \not\models \phi_h$.
Since,
$\Diamond^+ t \rightarrow \Box^+ t$
is valid on $\mathfrak F_\alpha$ by construction of $R^+$,
it suffices to exhibit a valuation $V$ and a world $w$ such that
$
\mathfrak F_\alpha,V,w \not\models
\Diamond(\Box^- t \wedge \Box^+\neg t).
$

Define a valuation $V$ on $\mathfrak F_\alpha$ by
$p_i\in V(t)$ iff $I(P_i)=\mathsf{true}$.
Fix any world $w$.
By construction, the $R$-successors of $w$ are exactly the clause nodes.
Hence
$
\mathfrak F_\alpha,V,w \models
\Diamond(\Box^- t \wedge \Box^+\neg t)
$
if and only if there exists a clause node $c_k$ such that
\[
\mathfrak F_\alpha,V,c_k \models
\Box^- t \wedge \Box^+\neg t.
\]
Let
$
C_k=\neg P_{i_1}\vee\dots\vee\neg P_{i_\ell}\vee P_h.
$
Because $I\models C_k$, at least one disjunct is true.
If $I(P_h)=\mathsf{true}$, then $\mathfrak F_\alpha,V,c_k\models\Diamond^+ t$.
Otherwise some $I(P_{i_r})=\mathsf{false}$, and $\mathfrak F_\alpha,V,c_k\models\Diamond^-\neg t$.
In either case,
$\mathfrak F_\alpha,V,c_k\not\models\Box^- t \wedge \Box^+\neg t$.

Thus no world satisfies
$\Diamond(\Box^- t \wedge \Box^+\neg t)$ under $V$,
and so $\mathfrak F_\alpha\not\models\phi_h$.

\medskip
Therefore the reduction is correct.
Since \textsc{Horn-UNSAT} is $\mathbf{P}$-complete,
the problem of deciding whether a finite frame validates $\phi_h$ is $\mathbf{P}$-hard under logspace reductions.
\end{proof}

Next, we exhibit an $\mathrm{FO+monLFP}$ sentence defining the same class of finite frames as $\phi_h$.
More precisely, we will show that over finite frames,
$\phi_h$ is equivalent to the $\mathrm{FO+monLFP}$ sentence
\[
\forall y\,\neg\mathrm{HornSat}(y)\ \wedge\ \mathrm{Func}(R^+),
\]
where $\mathrm{Func}(R^+)$ expresses that $R^+$ is functional, i.e., \[
\mathrm{Func}(R^+):=\forall x\forall y\forall z\,((R^+(x,y)\wedge R^+(x,z))\to y=z),
\]
and
$\mathrm{HornSat}(y)$ is a monadic least fixed-point formula expressing that the Horn formula induced by $y$ is satisfiable.
We construct this sentence step by step.
To this end, we begin by introducing some key constructions.

Fix a finite frame $\mathfrak F=(W,R,R^+,R^-)$ in which $R^+$ is functional.
For $a\in W$, let $\mathrm{Succ}_R(a)$, $\mathrm{Succ}_{R^-}(a)$, and $\mathrm{Succ}_{R^+}(a)$ denote the sets of $R$-, $R^-$-, and $R^+$-successors of $a$, respectively.
By functionality of $R^+$, $\mathrm{Succ}^+(a)$ has at most one element.
We further define
\[
\mathrm{Succ}^*(a):=\mathrm{Succ}^-(a)\cup \mathrm{Succ}^+(a).
\]

\begin{defi}\label{def:horn-formula}
Let $\mathfrak F=(W,R,R^-,R^+)$ be a finite frame in which $R^+$ is functional,
and let $a\in W$.
Enumerate the set
$
\bigcup_{c\in \mathrm{Succ}_R(a)} \mathrm{Succ}^*(c)
$
as $\{p_1,\dots,p_m\}$, and introduce the corresponding propositional variables
$\{P_1,\dots,P_m\}$.
Enumerate $\mathrm{Succ}_R(a)$ as $\{c_1,\dots,c_n\}$.
For each $i\in\{1,\dots,n\}$, define the Horn clause
\[
C_i
\ :=\
\left(
\bigvee_{p_j\in \mathrm{Succ}^-(c_i)} \neg P_j
\right)
\ \vee\
\left(
\bigvee_{p_j\in \mathrm{Succ}^+(c_i)} P_j
\right).
\]
The Horn formula associated with $\mathfrak F,a$ is defined as
\[
\mathrm{Horn}_{\mathfrak F,a}
\ :=\
\bigwedge_{i=1}^{n} C_i.
\]
\end{defi}

Intuitively, each $R$-successor $c$ of $a$ is thought of as a Horn clause: its $R^-$-successors yield the negative literals,
and its (unique, if any) $R^+$-successor yields the positive literal.
\medskip

% ----------  FO+LFP ingredient  ----------
Now we express the standard polynomial-time forward-chaining algorithm for Horn-SAT problem of formula $\mathrm{Horn}_{\mathfrak F,a}$ within FO+monLFP.
\begin{defi}
 We define the monadic LFP formula $\mathrm{HornCl}(y,z)$ by
\[
\mathrm{HornCl}(y,z) := [\mathrm{LFP}_{T,x}\,\alpha(T,x,y)](z),
\]
where 
\[
\alpha(T,x,y):=\exists u \Bigl(R(y,u)\wedge R^+(u,x)\wedge \forall v\,(R^-(u,v)\to T(v))\Bigr).
\]
\end{defi}

By construction, the predicate $\mathrm{HornCl}(y,z)$ computes precisely the assignment produced by the standard forward-chaining algorithm for Horn-SAT applied to $\mathrm{Horn}_{\mathfrak F,y}$.
We can therefore use this fixed-point predicate to express satisfiability of
$\mathrm{Horn}_{\mathfrak F,y}$ 
within $\mathrm{FO+monLFP}$.

\begin{defi}
We define
\[
\mathrm{HornSat}(y):=
\forall x\Bigl(
R(y,x) \to
\bigl(
\exists z\,(R^-(x,z)\wedge \neg\mathrm{HornCl}(y,z))
\ \vee\ 
\exists z\,(R^+(x,z)\wedge \mathrm{HornCl}(y,z))
\bigr)
\Bigr).
\]
\end{defi}

Intuitively, $\mathrm{HornSat}(y)$ states that for every clause $C_x$ induced by an $R$-successor $x$ of $y$, either some negative literal of $C_x$ is satisfied or its positive literal is satisfied, and hence $C_x$ itself is satisfied.

The following lemma makes this correspondence precise.
\begin{lem}\label{lem:chi-horn-modal}
Let $\mathfrak F=(W,R,R^+,R^-)$ be a finite frame such that $R^+$ is functional,
and let $a\in W$.
Then the following are equivalent:
\begin{enumerate}
\item $\mathfrak F \models \mathrm{HornSat}(a)$
\item $\mathrm{Horn}_{\mathfrak F,a}$ is satisfiable
\item there exists a valuation $V$ for the propositional letter $t$ such that
$
\mathfrak F,V,a\models \Box\bigl(\Diamond^- \neg t\vee \Diamond^+ t\bigr).
$
\end{enumerate}
\end{lem}

\begin{proof}

$(1)\Rightarrow(3)$:
Assume $\mathfrak F\models \mathrm{HornSat}(a)$.
Then, for every $c\in \mathrm{Succ}_R(a)$, either
\[
\mathfrak F \models \exists z\,(R^-(c,z)\wedge \neg\mathrm{HornCl}(a,z))
\quad\text{or}\quad
\mathfrak F\models \exists z\,(R^+(c,z)\wedge \mathrm{HornCl}(a,z)).
\]

Define a valuation $V$ for the propositional letter $t$ by
$
V(t) := \{\,p\in W \mid \mathfrak F\models \mathrm{HornCl}(a,p)\,\}.
$
Then $
\mathfrak F,V,c\models \Diamond^- \neg t\ \vee\ \Diamond^+ t.
$
Since $c\in \mathrm{Succ}_R(a)$ was arbitrary, it follows that
$
\mathfrak F,V,a\models \Box\bigl(\Diamond^- \neg t\vee\Diamond^+ t\bigr).
$

$(3)\Rightarrow(2)$:
Assume that there exists a valuation $V$ such that
$
\mathfrak F,V,a\models \Box\bigl(\Diamond^- \neg t\vee\Diamond^+ t\bigr).
$
Define a propositional assignment $\nu$ for the variables of
$\mathrm{Horn}_{\mathfrak F,a}$ by
\[
\nu(P_j)=\mathsf{true}
\quad\text{iff}\quad
\mathfrak F,V,p_j\models t
\quad\text{iff}\quad
p_j\in V(t).
\]
% where $P_j$ corresponds to the node $p_j$.
We show that $\nu$ satisfies every clause $C_i$.
Fix $c_i\in \mathrm{Succ}_R(a)$.
By the assumption, we have
\[
\mathfrak F,V,c_i\models \Diamond^- \neg t \vee \Diamond^+ t.
\]
If $\mathfrak F,V,c_i\models \Diamond^- \neg t$, choose $p_j$ such that
$R^-(c_i,p_j)$ and $\mathfrak F,V,p_j\models \neg t$.
Then $\nu(P_j)=\mathsf{false}$, so the negative literal $\neg P_j$ is true, and since $\neg P_j$ occurs in $C_i$, the clause $C_i$ is satisfied.
Otherwise $\mathfrak F,V,c_i\models \Diamond^+ t$.
Choose $p_j$ such that $R^+(c_i,p_j)$ and $\mathfrak F,V,p_j\models t$.
Then $\nu(P_j)=\mathsf{true}$, so the (unique, if any) positive literal $P_j$ is true, 
and since $P_j$ occurs in $C_i$, the clause $C_i$ is satisfied.

Thus every clause $C_i$ is satisfied by $\nu$,
and hence $\mathrm{Horn}_{\mathfrak F,a}$ is satisfiable.

\smallskip

$(2)\Rightarrow(1)$:
Assume $\mathrm{Horn}_{\mathfrak F,a}$ is satisfiable.
We define
\[
V_{\mathfrak F,a}:=\{p\in W\mid \mathfrak F\models \mathrm{HornCl}(a,p)\}.
\] 
The fixed point $V_{\mathfrak F,a}$ determines a propositional assignment for the variables
$\{P_1,\dots,P_m\}$ that occur in
$\mathrm{Horn}_{\mathfrak F,a}$.
More precisely, we define a propositional assignment
$
\nu : \{P_1,\dots,P_m\} \to \{\mathsf{true},\mathsf{false}\}
$
by
\[
\nu(P_j)=\mathsf{true}\ \text{ iff }\ p_j \in V_{\mathfrak F,a}.
\]
By construction of the least fixed point, $V_{\mathfrak F,a}$ coincides with the set of variables made $\mathsf{true}$ by the standard polynomial-time algorithm for Horn-SAT.
Indeed, $\mathrm{HornCl}(a,x)$ describes the iterative steps of the algorithm, and at each stage $k$ the set $V_k$ corresponds to the assignment maintained by the algorithm.
Hence $\nu$ satisfies every clause $C_i$ of
$\mathrm{Horn}_{\mathfrak F,a}$.

Thus for every clause $C_i$, either
$\nu(P_j)=\mathsf{false}$ for some negative literal $\neg P_j$ in $C_i$,
or $\nu(P_j)=\mathsf{true}$ for the (unique, if any) positive literal $P_j$ in $C_i$.
Equivalently, either $p_j\notin V_{\mathfrak F,a}$ for some $p_j\in \mathrm{Succ}^-(c_i)$,
or $p_j\in V_{\mathfrak F,a}$ for some
$p_j\in \mathrm{Succ}^+(c_i)$.

Hence for every $c_i\in \mathrm{Succ}_R(a)$, either
\[\mathfrak F \models \exists z\,(R^-(c_i,z)\wedge \neg\mathrm{HornCl}(a,z))
\quad\text{or}\quad
\mathfrak F\models \exists z\,(R^+(c_i,z)\wedge \mathrm{HornCl}(a,z)).
\]
Thus $\mathfrak F\models \mathrm{HornSat}(a)$.
\end{proof}

Lemma \ref{lem:chi-horn-modal} immediately yields the following.
\begin{prop}\label{prop:chi-func-modal}
Let $\mathfrak F=(W,R,R^+,R^-)$ be a finite frame.
Then the following are equivalent:
\begin{enumerate}
\item $\mathfrak F \models \forall y\,\neg\mathrm{HornSat}(y)
      \ \wedge\ \mathrm{Func}(R^+)$;
\item $\mathfrak F \models
      \Diamond(\Box^- t \wedge \Box^+ \neg t)
      \wedge
      (\Diamond^+ t \to \Box^+ t)$.
\end{enumerate}
\end{prop}

From Theorem~\ref{thm:horn-p}, Proposition~\ref{prop:chi-func-modal},
and the fact that $\mathrm{FO+monTC}$-definable classes of finite
structures lie in $\mathbf{NL}$, we obtain the following corollary.
\begin{cor}\label{cor:phi-not-montc}
The modal formula
$
\Diamond(\Box^- t \wedge \Box^+ \neg t)
\wedge
(\Diamond^+ t \to \Box^+ t)
$
is definable in $\mathrm{FO+monLFP}$ but not in $\mathrm{FO+monTC}$
over finite frames, unless $\mathbf{NL}=\mathbf{P}$.
\end{cor}

\medskip
We conclude with some questions for further research. The preceding results show that, under standard complexity-theoretic assumptions, the modal hierarchy over finite frames does not collapse.
A natural question is whether these separations can be proved without relying on complexity-theoretic assumptions. 
Another question is how such hierarchies behave over restricted classes of finite frames.
For example, on transitive frames,
the McKinsey Axiom is equivalent to the first-order formula that expresses ``atomicity'':
$
\forall x \,\exists y\,
(Rxy \land Ryy \land \forall z\,(Ryz \rightarrow z = y)).
$
This suggests asking how the hierarchy of modal definability behaves over finite transitive frames.
Also, over finite trees it is known that
$
\mathrm{FO}
\subsetneq
\mathrm{FO+monTC}
\subsetneq
\mathrm{FO+monLFP}
\equiv
\mathrm{MSO}
$
\cite{CateSegoufin2010}.

Clearly we are only at the start of a whole area of definability questions here.

\section{Conclusion}

We have investigated transfer of modal preservation results from the setting of arbitrary structures to the finite, both for modal models carrying a valuation and for relational  frames. The results of our systematic survey were partly positive and partly negative, in terms of failures of classical preservation and definability properties. 

At the end of Section~\ref{sec:FO}, we have suggested that this taxonomy of success and failure need not be the end of the story of what makes the finite realm special. We noted how there might be interesting new preservation results still to be found that exploit the special nature of finite structures. In fact, existing counter-examples in this paper and the broader literature might contain cues for this, when read with this purpose in mind. 

A more concrete positive outcome of our analysis occurred in Section~\ref{sec:GbTh}. Here we showed that the hierarchy of correspondence languages known from general modal model theory does not collapse in the finite, and in fact, modal axioms now get content that can be investigated using known results from computational complexity theory that were not available in the modal theory of arbitrary relational  structures.

Our final thought is as follows. In addition to transfer of \emph{results} from the general case to the finite, there is also a broader issue of transfer of \emph{proof techniques}. For instance, our analysis of the McKinsey Axiom in the finite was inspired by the proof of non-first-order definability for the general case of some modal logics in \cite{vanBenthem1985}. And it is also known from general modal model theory that proofs of modal results often follow by analyzing known proofs for first-order logic and replacing appeals to isomorphism or potential isomorphism by uses of bisimulation~\cite{deRijkePhD}. 
We leave it as an interesting line of thought how such broader forms of transfer of proof strategies might work for the themes raised in this paper.

\newpage

\bibliographystyle{alphaurl}
\bibliography{bib}

\end{document}